\documentclass[peerreview]{IEEEtran}
\usepackage[margin=1in]{geometry}
\newgeometry{top=1in,bottom=1in,right=1in,left=1in}
\usepackage[utf8]{inputenc} 
\usepackage[T1]{fontenc}
\usepackage{url}
\usepackage{ifthen}
\usepackage[cmex10]{amsmath}
\usepackage{amsmath,bm}
\usepackage{amsthm}
\usepackage{bbm}
\usepackage{amssymb}
\usepackage{float}
\usepackage{mathtools}
\usepackage{graphicx}
\usepackage{xcolor}
\usepackage[english]{babel}
\usepackage{tabularx}
\usepackage{multirow}
\usepackage{stfloats}
\usepackage[normalem]{ulem}
\usepackage{cite}
\usepackage[ruled,vlined]{algorithm2e}
\SetKwInput{KwInput}{Input}                % Set the Input
\SetKwInput{KwOutput}{Output}              % set the Output
\usepackage{algorithmic}
\usepackage{todonotes}
\usepackage{mleftright}\mleftright

\usepackage{graphicx}

\usepackage{tikz}\usetikzlibrary{shapes.multipart}
\usetikzlibrary{positioning}
\tikzset{font={\fontsize{9pt}{12}\selectfont}}
\tikzset{>=latex}
\usetikzlibrary{calc}
\usetikzlibrary{arrows}
\usetikzlibrary{external}
\usetikzlibrary{patterns}
\usetikzlibrary{fit}
\usepackage{changes}

\DeclareMathOperator{\wt}{wt}

\newtheorem{theorem}{Theorem$\!$}
\newtheorem{lemma}{Lemma$\!$}
\newtheorem{claim}{Claim$\!$}
\newtheorem{corollary}{Corollary$\!$}
\newtheorem{proposition}{Proposition$\!$}
\theoremstyle{definition}
\newtheorem{construction}{Construction$\!$}

\newtheorem{examplex}{Example$\!$}
\newenvironment{example}
  {\pushQED{\qed}\renewcommand{\qedsymbol}{$\triangle$}\examplex}
  {\popQED\endexamplex}
\usepackage{graphicx}

\usepackage{xcolor}

\newcommand{\bcomment}[1]{{\leavevmode\color{blue}#1}}

\newcommand{\mA}{\mathsf{A}}
\newcommand{\mC}{\mathsf{C}}
\newcommand{\mG}{\mathsf{G}}
\newcommand{\mT}{\mathsf{T}}

\newcommand{\cB}{\mathcal{B}}
\newcommand{\cC}{\mathcal{C}}
\newcommand{\cD}{\mathcal{D}}
\newcommand{\cE}{\mathcal{E}}

\newcommand{\cH}{\mathcal{H}}

\newcommand{\cM}{\mathcal{M}}

\newcommand{\cP}{\mathcal{P}}

\newcommand{\cR}{\mathcal{R}}
\newcommand{\cS}{\mathcal{S}}

\newcommand{\mybold}[1]{\bm{#1}}

\newcommand{\bb}{{\mybold{b}}}

\newcommand{\bh}{{\mybold{h}}}

\newcommand{\bp}{{\mybold{p}}}

\newcommand{\bs}{{\mybold{s}}}

\newcommand{\bu}{{\mybold{u}}}
\newcommand{\bv}{{\mybold{v}}}
\newcommand{\bw}{{\mybold{w}}}
\newcommand{\bx}{{\mybold{x}}}
\newcommand{\by}{{\mybold{y}}}
\newcommand{\bz}{{\mybold{z}}}

\newcommand{\balpha}	{\mybold{\alpha}}

\newcommand{\bpi}		{\mybold{\pi}}

\newcommand{\bphi}		{\mybold{\phi}}

\newcommand{\bpsi}		{\mybold{\psi}}

\DeclareMathOperator{\VT}{VT}
\DeclareMathOperator{\pr}{Pr}
\DeclareMathOperator{\prj}{Prj}

\newcommand{\pll}{\mathit{PLL}}
\newcommand{\plld}{\mathit{PLL}\text{-}}

\usepackage{calligra} \DeclareMathAlphabet{\mathcalligra}{T1}{calligra}{m}{n} \DeclareFontShape{T1}{calligra}{m}{n}{<->s*[2.2]callig15}{}

\newcommand{\sr}{\ensuremath{\mathcalligra{r}}}

\begin{document}

\title{Non-binary Codes for \\ Correcting a Burst of at Most $t$ Deletions}

\author{
  \IEEEauthorblockN{Shuche Wang\IEEEauthorrefmark{1}, Yuanyuan Tang\IEEEauthorrefmark{2}, Jin Sima\IEEEauthorrefmark{3}, Ryan Gabrys\IEEEauthorrefmark{4} and Farzad Farnoud\IEEEauthorrefmark{2}}\\
  \IEEEauthorblockA{\IEEEauthorrefmark{1}
                    Institute of Operations Research and Analytics, 
                    National University of Singapore, 
                    \texttt{shuche.wang@u.nus.edu}} \\
  \IEEEauthorblockA{\IEEEauthorrefmark{2}
                    Electrical \& Computer Engineering, 
                    University of Virginia, 
  \texttt{\{yt5tz,farzad\}@virginia.edu}}\\ 
  \IEEEauthorblockA{\IEEEauthorrefmark{3}%
                    Electrical \& Computer Engineering, 
                    University of Illinois Urbana-Champaign, \texttt{jsima@illinois.edu}}\\
    \IEEEauthorblockA{\IEEEauthorrefmark{4}
  Calit2, University of California-San Diego, U.S.A., \texttt{rgabrys@eng.ucsd.edu} }
                    
  \thanks{This work was supported in part by NSF grants under grant nos.~1816409 and~1755773. This paper was presented in part at the 2021 IEEE International Symposium on Information Theory (ISIT) in 2021~\cite{wang2021non} and 58th Allerton Conference on Communication, Control, and Computing~\cite{Wang2022permutation}.}        
}

\maketitle

\begin{abstract}
The problem of correcting deletions has received significant attention, partly because of the prevalence of these errors in DNA data storage. In this paper, we study the problem of correcting a consecutive burst of at most $t$ deletions in non-binary sequences. We first propose a non-binary code correcting a burst of at most 2 deletions for $q$-ary alphabets. Afterwards, we extend this result to the case where the length of the burst can be at most $t$ where $t$ is a constant. Finally, we consider the setup where the sequences that are transmitted are permutations. The proposed codes are the largest known for their respective parameter regimes.
\end{abstract}

%\tableofcontents

\section{Introduction}

Codes correcting insertions/deletions have garnered significant recent interest due to their relevance in many applications such as storage~\cite{chee2019burst,yazdi2015dna}, communication systems~\cite{dolecek2007using} and file synchronization~\cite{venkataramanan2010interactive}. Constructing codes in the insertion/deletion metric is a notoriously difficult problem whose origins date back to at least the 1960s~\cite{levenshtein1966binary}. One of the challenges under this setup is that deletions seem to be more destructive in nature than substitutions as only a relatively small number of insertions/deletions can cause the transmitted and received
sequences to be vastly different under the Hamming metric.

One of the motivations for the current work is the recent emergence of DNA-based storage systems \cite{yazdi2015dna}. Unlike traditional information systems whose dominant source of errors stems from substitutions, data stored in DNA is often corrupted by bursts of insertions and deletions~\cite{lee2019terminator}. Motivated by this connection, the current work focuses on the development of non-binary codes capable of correcting a consecutive burst of deletions. 

Previous works have studied the problem of constructing codes over binary alphabets, and optimal codes exist for many setups of interest. In what was perhaps the earliest work on the subject, Levenshtein constructed a code capable of correcting a burst of length at most two that had redundancy $\log n + 1$ \cite{levenshtein1967asymptotically}. 
%We will expand upon some of the tools developed to tackle this problem in Section~\ref{sec:bicor2} where we will construct a nearly optimal code capable of correcting a burst of at most $2$ deletions. 
In \cite{schoeny2017codes}, Schoeny et al. proposed burst deletion correcting codes for the setup where the length of the burst is exactly $t$ and the deletions are consecutive. In \cite{lenz2020optimal}, Lenz and Polyanskii presented codes that correct consecutive bursts of deletions of length at most $t$ that required only $\log n + O(\log \log n)$ bits of redundancy. The best-known systematic codes can be found in \cite{sima2020syndrome}. A summary of these results is included in Table~\ref{tab:binary}.

Unlike the previously mentioned works, the goal in this paper is to construct low redundancy \textit{non-binary} codes capable of correcting a burst of length \textit{at most} $t$ where $t$ is a constant. In this work, we also consider the setup where the non-binary sequences that comprise our code are permutations. Our main results, which  are highlighted in Table~\ref{tab:nonbinary}, are the following:
\begin{enumerate}
    \item We present a simpler proof for Levenshtein's binary code~\cite{levenshtein1967asymptotically}, based on the well-known Varshamov-Tenengol'ts constraint. We then show this proof can be used for constructing a code with redundancy at most $\log n+2\log q +1$ for correcting an induced burst of 2 deletions in alternating sequences\footnote{Induced deletions occur in sequences where every two adjacent symbols are different. This setup is motivated by the recently proposed terminator-free synthesis of DNA sequences~\cite{lee2019terminator}.}. 
    \item We construct non-binary codes for correcting a burst of at most $t$ deletions for $q$-ary alphabets that has redundancy $\log n+O(\log q\log\log n)$.
    \item Using ideas developed in the context of non-binary codes, we present a permutation code for correcting a burst of at most $t$ deletions that has redundancy $\log n+O(\log\log n)$.
\end{enumerate}
To the best of the authors' knowledge, the codes presented here are the largest known codes for each of the parameter regimes under consideration. We note that result (3), which appears in Section~\ref{sec:per}, was simultaneously and independently derived in \cite{sun2022improved}. Since our approach uses a different technique than the one from \cite{sun2022improved}, it may be of independent interest.

%We note that the results in Table~\ref{tab:nonbinary} listed for permutation codes were recently improved and  

\begin{table}[H]
\begin{minipage}{.5\linewidth}
\caption{Related works for binary codes}
\label{tab:binary}
\centering
\begin{tabular}{|l|l|l|}
\hline
            & Size of burst & Redundancy   \\ \hline
%Levenshtein, 1967~
\cite{levenshtein1967asymptotically}    & $\leq 2$     & $\log n+1$      \\ \hline
%Schoeny et al., 2017~
\cite{schoeny2017codes} & $=t$     & $\log n+O(\log\log n)$  \\ \hline
%Schoeny et al., 2017~
\cite{schoeny2017codes}  & $\leq t$     & $(t-1)\log n+O(\log\log n)$  \\ \hline
%Lenz et al., 2020~
\cite{lenz2020optimal}    & $\leq t$     & $\log n+O(\log\log n)$ \\ \hline
%Sima et al., 2020~
\cite{sima2020syndrome}    & $\leq t$     & $4\log n+o(\log n)$          \\ \hline
\end{tabular}
\end{minipage}%
    \begin{minipage}{.5\linewidth}
    \caption{Related works for non-binary (NB)/permutation codes}
\label{tab:nonbinary}
\begin{tabular}{|l|l|l|}
\hline
            & Size of burst & Redundancy   \\ \hline
%Schoeny et al., %2017~
\cite{schoeny2017novel} (NB)  & $=t$     & $\log n+O(\log\log n+\log q)$  \\ \hline
\cite{chee2019burst} (P)  & $=t$     & $2\log n$  \\ \hline
\cite{chee2019burst} (P)      & $\leq t$     & $2t\log n$  \\ \hline
\bf{This Work} (NB)     & $\leq 2$     & $\log n+O(\log q\log\log n)$  \\ \hline
\bf{This Work} (NB)      & $\leq t$     & $\log n+O(\log q\log\log n)$\\   \hline
\bf{This Work} (P)       & $\leq t$    & $\log n+O(\log\log n)$\\   \hline
\end{tabular}
\end{minipage}
\end{table}

The remainder of this article is outlined as follows. Section~\ref{sec:notation} presents the notations and two well-known deletion correcting codes used throughout this paper as well as some preliminary results. Section~\ref{sec:bicor2} gives an alternative proof of the Levenshtein code and a code for correcting deletions in alternating sequences is proposed based on this proof. In Section~\ref{sec:noncor2}, when $q$ is even, we construct a non-binary code for correcting a burst of at most 2 deletions for $q$-ary alphabets with redundancy $\log n+O(\log q\log\log n)$. In Section~\ref{sec:noncort}, when $q$ is even, we construct a non-binary code for correcting a burst of at most $t$ deletions for $q$-ary alphabets with redundancy $\log n+O(\log q\log\log n)$. Section~\ref{sec:per} proposes a permutation code for correcting a burst of at most $t$ deletions with redundancy $\log n+O(\log\log n)$. Finally, Section~\ref{sec:conclusion} concludes the paper.

\section{Notation and Preliminaries}\label{sec:notation}
We now describe the notations used throughout this paper. Let $\Sigma_q$ denote a finite alphabet of size $q$ and $\Sigma_q^n$ represent the set of all sequences of length $n$ over $\Sigma_q$. Without loss of generality, we assume $\Sigma_q=\{0,1,\dotsc,q-1\}$. For ease of notation, we will denote the set $\{0,1,\ldots, m-1\}$ as $[[m]]$ and the set $\{1,2,\ldots,m\}$ as $[m]$. For two integers $i<j$, let $[i,j]$ denote the set $\{i,i+1,i+2, \ldots, i+j\}$.

We write sequences with bold letters, such as $\bu$ and their elements with plain letters, e.g., $\bu=u_1\dotsm u_n$ for $\bu\in \Sigma_q^n$. For functions, if the output is a sequence, we also write them with bold letters, such as $\bphi(\bu)$. The $i$th position in $\bphi(\bu)$ is denoted $\phi(\bu)_i$. We typically use $\bu$ for non-binary and $\bx$ for binary sequences.  The length of the sequence $\bx$ is denoted $|\bx|$ and  $\bx_{[i,j]}$ denotes the substring beginning  at index $i$ and ending at index $j$, inclusive. A run is a maximal substring consisting of identical symbols. The weight $\wt(\bu)$ of a sequence $\bu$ represents the number of non-zero symbols in it. A sequence $\bu$ is said to have period 2 if $u_{i}=u_{i+2}$ for $i \in [|\bu|-2]$. Define $L(\bu,2)$ as the length of the longest substring of $\bu$ with period $2$. %The period 2 of sequence $\bx$ means the symbols in $\bx$ such that  

A \emph{burst of $t$ deletions} deletes  $t$ consecutive symbols from $\bu=u_1\dotsm u_n$ starting in position $i+1$ resulting in $\bu'= u_1\dotsm u_{i}u_{i+t+1}\dotsm u_n$ for $\bu\in\Sigma_q^n$ (we usually use $\bu'$ to denote the sequence resulting from deleting symbols from $\bu$). For shorthand, let $D_t(\bu) \subseteq \Sigma_q^{n-t}$ denote the set of all sequences possible given that $t$ deletions occur to $\bu$ and similarly let $D_{\leq t}(\bu) = D_1(\bu) \cup D_2(\bu) \cup \cdots \cup D_t(\bu)$. The size of a code $\cC\subseteq \Sigma_q^n$ is denoted $|\cC|$ and its redundancy is defined as $\log (q^n/|\cC|)$, where all logarithms in this paper are to the base 2. We say that a code $\cC \subseteq \Sigma_q^n$ is a \emph{$t$-burst-error-correcting code} if for two distinct $\bu, \bv \in \cC$, $D_{\leq t}(\bu) \cap D_{\leq t}(\bv)  = \emptyset$.

For the permutation code, let $n$ be a positive integer and $\cS_n$ be the set of all permutations on the set $[n]$. Denote $\bpi=(\pi_1,\pi_2,\dotsc,\pi_{n})\in \cS_n$ as a permutation code with length $n$. A burst of at most $t$ deletions deletes at most $t$ consecutive symbols from the permutation code $\bpi$, leading to $\bpi'=(\pi_1,\pi_2,\dotsc,\pi_i,\pi_{i+t'+1},\dotsc,\pi_n)$, where $t'\leq t$. The size of a permutation code $\bpi\in\cS_n$ is denoted $|\bpi|$ and its redundancy is defined as $\log (n!/|\bpi|)$.

The following codes will be of use in the paper. First for $\bw\in\Sigma_q^n$, define the VT syndrome as $\VT(\bw) = \sum_{i=1}^{n} i w_i$. Furthermore, define $\bphi:\Sigma_q^n\rightarrow\Sigma_2^n$ as
\begin{equation*}
    \phi(\bu)_i=\left\{\begin{array}{ll}
1, & \text { if } u_i>u_{i-1} \\
0, & \text { if } u_i \le u_{i-1}
\end{array}\right.
\end{equation*}
for $i\ge 2$ and $\phi(\bu)_1=1$.

\begin{theorem}[Varshamov-Tenengol'ts (VT) code \cite{sloane2000single}]
For integers $n$ and $a \in [[n+1]]$, 
\begin{equation}
    \VT_a(n)=\left\{ \bx\in\Sigma_2^n:\VT(\bx)\equiv a\bmod (n+1) \right\}
\end{equation}
is a $1$-burst-error-correcting code.
\end{theorem}

\begin{theorem}[Tenengol'ts code \cite{tenengolts1984nonbinary}]
For integers $n$, $a \in [[n]]$, and $b \in [[q]]$, the code
\begin{equation}
    \cC_T(a,b,n)= \left \{
    \bu\in\Sigma_q^n:\VT(\bphi(\bu))%\sum\nolimits_{i=1}^{n}i\phi(\bu)_i
    \equiv a\bmod n, \quad \sum\nolimits_{i=1}^{n}u_i\equiv b\bmod q \right\}
\end{equation}
is a $1$-burst-error-correcting code.
\end{theorem}

%Next, we consider the maximum size of a $t$-burst-error-correcting code. 

Let $\cM_t(n) \subseteq \Sigma_q^n$ be a $t$-burst-error-correcting code of maximum cardinality. Theorem~\ref{thm:nonasymptotic} provides a non-asymptotic upper bound on the size of $\cM_t(n)$ using the linear programming technique from \cite{schoeny2017codes}. A detailed proof is included in Appendix~A.

\begin{theorem}\label{thm:nonasymptotic}
For $n > t$ and $t | n$, we have
\begin{equation}\label{eq:nonasym_bound}
    \left| \cM_t(n) \right|\leq \frac{q^{n-t+1}-q^t}{(q-1)(n-2t+1)}
\end{equation}
\end{theorem}

From Table~\ref{tab:binary} and \ref{tab:nonbinary}, we can see that the redundancies of our constructions are only off from the minimum redundancy by a factor of at most roughly $\log q\log\log n$. %\textcolor{red}{NOTE: We should be careful here because the statement of the proof requires that $t |n$.}

Motivated by applications to storage systems with larger alphabets, we will also be interested in codes over permutations. Let $\cS_n$ be the set of all permutations on the set $[n]$ and $\cM_t^P(n) \subseteq \cS_n$ be a $t$-burst-error-correcting code of maximum cardinality. Theorem~\ref{thm:maxsize_permutation} provides a upper bound on the size of $\cM_t^P(n)$. 

\begin{theorem}\label{thm:maxsize_permutation}
(\cite[Theorem 1]{chee2019burst}) Let $n>t$ be positive integers. Then, 
\begin{equation*}
     \left| \cM_t^P(n) \right|\leq \frac{n!}{t!(n-t+1)}.
\end{equation*}
\end{theorem}
From Table~\ref{tab:nonbinary}, notice that the redundancy of our construction for permutation codes is only off from the minimum redundancy by at most $O(\log\log n)$.

\section{Binary code correcting a burst of at most 2 deletions}\label{sec:bicor2}
In this section, we first describe the Levenshtein code \cite{levenshtein1967asymptotically} in Subsection~\ref{subsec:Leven}, which can correct a burst of at most 2 deletions in binary sequences. In Section~\ref{subsec:alterproof_leven}, we provide an alternative formulation of the Levenshtein code, and using this formulation, we prove the correctness of the construction by describing the decoding algorithm. Our proof is in some ways simpler than Levenshtein's original proof, and it is similar to the well-known proof of the VT code~\cite{sloane2000single}. It will also enable us to construct a code for correcting deletions in alternating sequences in Section~\ref{subsec:alternating}. 
%\subsection{Levenshtein code}

\subsection{Levenshtein binary codes for correcting at most 2 deletions}\label{subsec:Leven}

For a binary sequence $\bx$, let  $\sr(\bx)$ be the sequence whose $i$th element is the \emph{run index} of $x_i$ in $\bx$, where the indexing starts from 0. Then, define the syndrome $\VT^{(r)}({\bx})$ of the sequence $\sr(\bx)$ as
\begin{equation}\label{eq:M}
    \VT^{(r)}({\bx})=\sum\nolimits_{i=1}^{n} {\sr(\bx)}_{i}.
\end{equation}
For example, if ${\bx}=01110100$, then $\sr(\bx)=01112344$, and $\VT^{(r)}(\bx)=1\times3+2\times1+3\times1+4\times2=16$.

\begin{theorem}[Levenshtein code~\cite{levenshtein1967asymptotically}]
For integers $n$ and $a \in [[2n]]$,  
\begin{equation}
    \cC_L(a,n)=\left \{\bx \in \Sigma_2^n: \VT^{(r)}(0\bx) \equiv a\bmod 2n \right \}
\end{equation}
is a $2$-burst-error-correcting code.
\end{theorem}

Using a simple averaging argument, it is straightforward to observe that $|\cC_L(a,n)| \geq \frac{2^{n-1}}{n}$ and so the redundancy of this code is at most $1 + \log n$ (for some $a \in [[2n]])$. For shorthand, we refer to the code $\cC_L(a,n)$ as the Levenshtein code.

\subsection{An alternative formulation of the Levenshtein code}\label{subsec:alterproof_leven}

Define $\bpsi(\bx) \in \Sigma_2^n$ to be the derivative of $\bx \in \Sigma_2^n$ so that
\begin{equation}\label{eq:psi}
    \psi(\bx)_i=\left\{
    \begin{array}{ll}
    x_{i}\oplus x_{i+1}, &  i=1,2,\dotsc,n-1 \\
    x_{n}, &  i=n
\end{array}
\right.
\end{equation}
where $a\oplus b$ denotes $(a+b)\bmod 2$. Levenshtein~\cite{levenshtein1967asymptotically} showed that  $\VT(\bpsi(\bx))\equiv -\VT^{(r)}(0\bx)\pmod{2n}$ for $\bx\in\Sigma_2^n$. This equality provides another way to prove the error-correcting capability of the code.
\begin{theorem}\label{thm:LC-alternative}
The code 
\(
\left\{\bx\in\Sigma_2^n: \VT(\bpsi(\bx))\equiv a\bmod 2n\right\}
\)
can correct a burst of at most 2 deletions.
\end{theorem}
\begin{proof}
For a codeword $\bx$, let $\bx'$ be obtained from $\bx$ after a burst of at most 2 deletions. Also, let $\by = \bpsi(\bx)$, $\by'= \bpsi(\bx')$, and $\Delta = \VT(\by)-\VT(\by')$. The error in $\bx$ can affect $\by$ in the following ways:
\begin{itemize}
    \item If the first one or two symbols of $\bx$ are deleted, then the first one or two symbols of $\by$ are deleted, respectively.
     \item If $x_i$ is deleted (where $i \in \{2,3,\ldots, n\}$), then $y_{i-1}y_{i}$ is replaced by $y_{i-1}\oplus y_{i}$. The possible error patterns in $\by$ are: $00\to 0, 01\to 1, 10\to 1, 11\to 0$. %Notice that the first three error patterns can be expressed in terms of a single deletion (for instance $01 \to 1$ is the result of deleting the symbol $0$). 
    \item If $x_{i}x_{i+1}, 2\le i\le n-1$, are deleted, then $y_{i-1}y_{i}y_{i+1}$ is replaced by $y_{i-1}\oplus y_{i}\oplus y_{i+1}$. The error patterns in $\by$ are: $000\to 0, 001\to 1, 010\to 1, 011\to 0,100\to 1, 101\to 0, 110\to 0, 111\to 1$. 
\end{itemize}
Let $R_1$ be the number of 1s on the right of the altered bits in $\by$ and $L_1$ on their left. Similarly, define $R_0,L_0$. Let $w$ denote the weight of $\by'$ and $p$ the index of the first altered bit. We show how $\by$ and thus $\bx$ can be recovered based on $\Delta$ and $w$, which are computable at the decoder. % \textcolor{red}{NOTE: Are we being too flimsy with the definition of $p$ and the notion of ``altered bits''?}

If one bit is deleted in $\bx$, the following cases may occur in $\by$ resulting in $\by'$. As can be seen below, the cases can be distinguished by comparing $\Delta$ and $w$ so that it is possible to recover $\by$ in each case.
\begin{enumerate}
    \item If $00 \to 0$, $01 \to 1$, or $10 \to 1$, then $\Delta=R_1\le w$. We recover $\by$ by inserting a 0 in the rightmost position before $\Delta$ 1s. 
    \item If the first bit is deleted from $\by$ resulting in $\by'$ and it is a 1, then $\Delta=R_1+1=w+1$. To correct the error, we prepend 1 to $\by'$.
    \item If $11 \to 0$, then $\Delta = 2p+1+R_1 =2 \left(L_0 + L_1 + 1 \right) + 1 + R_1 = 2L_0 + L_1 + w + 3\ge w+3$. As $2p+1+R_1$ is strictly increasing in $p$, there is a unique value of $p$ satisfying the equation and we can again recover $\by$.
\end{enumerate}

If two bits are deleted in $\bx$, one of the following cases occurs in $\by$. Similar to before, each case can be identified based on the comparison of $\Delta$ with $w$.
\begin{enumerate}
    \item If $010\to 1$ occurs, then $\Delta = 2R_1+1 < 2w$. We insert two 0s on both sides of the $(R_1+1)$th 1 from the right.
    \item If $001 \to 1$ or $100 \to 1$, then $\Delta = 2R_1\le 2w$. We insert two 0s on the left of the $R_1$th 1 from the right. Note that in this case, $\Delta$ is an even number whereas in 1) it was odd.
    \item If $10$ or $01$ is deleted from the beginning, then $\Delta=2R_1+1=2w+1$ or $\Delta = 2w+2$, respectively. Here 10 or 01 is prepended, depending on the parity of $\Delta$. 
    \item If $011 \to 0$ or $111 \to 1$, then  $\Delta= p + \left(p+1\right) + 2R_1$ $= \left(L_0 + L_1 + 1\right) + \left(L_0 + L_1 + 2 \right) + 2 R_1$ $= 2w + 3 + 2 L_0 \geq 2w + 3$. In this case, we insert 11 immediately after the $L_0$th 0, where $L_0 = (\Delta-2w-3)/2$. 
    \item If $101\to 0$ occurs, then $\Delta = 2p+2+2R_1\ge 2w+4$. We insert two 1s on the sides of the $(L_0+1)$th 0, where $L_0 = \Delta/2-w-2$.\qedhere
\end{enumerate}
\end{proof}

\subsection{Correcting an induced burst of 2 deletions in alternating sequences}\label{subsec:alternating}

\iffalse

\fi

One of the challenges facing DNA data storage is the high cost of synthesizing DNA sequences accurately. Recently, a new enzymatic method for parallel DNA synthesis was proposed in~\cite{lee2019terminator}, which could decrease the cost of synthesis but with a loss in accuracy. Specifically, the length of the runs cannot be easily controlled. One approach to address this challenge is to store information only in the identity of the symbols of the runs. In this approach, for example, $\mA\mA\mA\mC\mG\mG\mC\mC\mT$, $\mA\mC\mC\mG\mG\mC\mT\mT\mT$, and $\mA\mC\mG\mC\mT$ would be equivalent. In this case, we can consider the information being encoded in a sequence with no adjacent repeats of a symbol, such as $\mA\mC\mG\mC\mT$, which we term an \emph{alternating} sequence.

The second difficulty of this enzymatic synthesis method is the high probability of deletion~\cite{lee2019terminator}, which makes it possible for complete runs to be deleted. Over the space of alternating sequences, the deletion of a single symbol may manifest as a burst of more than 1 deletion. For example, the deletion $\mA\mC\mC\mG\mG\mC\mT\mT\mT\to \mA\mC\mC\mC\mT\mT\mT$ will be interpreted as $\mA\mC\mG\mC\mT\to \mA\mC\mT$. We call such a burst resulting from a single deletion an \emph{induced deletion}. Formally, a single induced deletion in an alternating sequence is a deletion that replaces a substring of the form $aba$ with $a$, where $a, b\in\Sigma$ and $a\neq b$. Induced deletions may be more difficult to handle for trace reconstruction and synchronization approaches given that they cause larger shifts in the sequence. In this subsection, we propose a non-binary code with redundancy $\log n+2\log q+1$ to correct an induced burst of length 2, which is very close to the bound shown in \eqref{eq:nonasym_bound}.

For sequences $\bv,\bw$, let $\bv\circ\bw=v_1w_1v_2w_2\dotsm$ be obtained by interleaving them. Define the odd and even subsequences of $\bv$ as $\bv^{o}=v_1v_3\dotsm$ and $\bv^{e}=v_2v_4\dotsm$, respectively. 

%\begin{construction}\label{const:cL}
%For integers $a \in [[2n]]$, $b \in [[q]]$, $c \in %[[q]]$, let
%\begin{equation*}
%    \cC_I(a,b,n) \buildrel \Delta \over = \left \{\bu\in \Sigma_q^n:
%    \VT(\bpsi(\bphi(\bu^o)\circ\bphi(\bu^e)))%\sum\nolimits_{i=1}^{n} i\psi(\bx)_i 
%    \equiv a \bmod  2n,\quad
%    \sum_{i}\bu^o_i%\sum\nolimits_{i=1}^{n/2} u_{2i-1} 
%    \equiv b \bmod q,%\\&\sum\nolimits_{i=1}^{n/2} u_{2i} 
%    \quad \sum_i\bu^e_i\equiv c \bmod q \right \}.
%\end{equation*}
%\end{construction}

\begin{theorem}\label{th:cInduce}
For integers $a \in [[2n]]$, $b \in [[q]]$, $c \in [[q]]$, 
\begin{equation*}
    \cC_I(a,b,n) = \left \{\bu\in \Sigma_q^n:
    \VT(\bpsi(\bphi(\bu^o)\circ\bphi(\bu^e)))%\sum\nolimits_{i=1}^{n} i\psi(\bx)_i 
    \equiv a \bmod  2n,\quad
    \sum_{i}\bu^o_i%\sum\nolimits_{i=1}^{n/2} u_{2i-1} 
    \equiv b \bmod q,%\\&\sum\nolimits_{i=1}^{n/2} u_{2i} 
    \quad \sum_i\bu^e_i\equiv c \bmod q \right \}.
\end{equation*}
can correct an induced deletion of size 2%, i.e., a burst of exactly 2 deletions caused by a single deletion
.
\end{theorem}
\begin{proof}
Let $\bu\in \cC_I(a,b,n)$ be a codeword and let $\bx = \bphi(\bu^o)\circ\bphi(\bu^e)$. Note that $x_i=1$ if $u_i>u_{i-2}$ or if $i\le 2$, and $x_i=0$ otherwise. Also, let $\bu'$ and $\bx'$ be the corresponding sequences after the deletions. We first characterize the changes in $\bx$. Suppose $u_{i+1}u_{i+2}$ are deleted, where $u_i = u_{i+2}$, for $2\le i\le n-3$. Then $x_{i+1}x_{i+2}x_{i+3}$ will change to $x'_{i+1}$ in $\bx$. Since $u_{i}=u_{i+2}$, we have $x_{i+2}=0$. Furthermore, note that if $i\le n-4$, then  $x_{i+4}=x'_{i+2}$ again because $u_i=u_{i+2}$, so we do not consider $x_{i+4}$ as part of the error pattern in $\bx$. Hence, the change in $\bx$ can be viewed as $x_{i+1}0x_{i+3}\rightarrow x'_{i+1}$, and the possible cases are $101\rightarrow 1$, $100\rightarrow 1$, $100\rightarrow 0$, $001\rightarrow 1$, $001\rightarrow 0$ and $000\rightarrow 0$. 

If the error is of the form $u_1u_2u_3\to u_3$, where $u_1=u_3$, then the change in $\bx$ will be of the form $110x_4\to11$. The three-bit patterns are then $101\to 1,100\to1$, both of which appear among the patterns above. 
If the error is of the form $u_{n-2}u_{n-1}u_{n}\to u_{n-2}$, where $u_n = u_{n-2}$, then the change in $\bx$ is of the form $x_{n-2}x_{n-1}0\to x_{n-2}$. The last three bits of $\bx$ change as $000\to0, 010\to 0, 100\to1$, or $110\to1$.

Having determined the changes in $\bx$, it can be then shown that the change in $\bpsi(\bx)$ is one of the following cases in Table~\ref{tab:induced_deletion}: deletion of 11, deletion of 00, $010\to1$, or $101\to0$. Then, similar to the proof of Theorem~\ref{thm:LC-alternative}, it can be shown that with the knowledge of $\VT(\bpsi(\bx))$, we can fix the errors in $\bpsi(\bx)$ and in turn the errors in $\bx$. From $\bx$, we find $\bx^o=\bphi(\bu^o)$ and $\bx^e=\bphi(\bu^e)$. 
Further, from $\sum_{i}\bu^o_i\equiv b \bmod q$ and $\sum_{i}\bu^e_i\equiv c \bmod q$, we can get the value of the deleted symbols in $\bu^o$ and $\bu^e$, respectively. At last, since $\bu'^o$, $\bx^o$ and the value of the deleted symbol are known, $\bu^o$ can be recovered via the decoding process for Tenengol'ts code $\cC_T(a,b,n)$. Also, $\bu^e$ can be recovered in the same way.
\end{proof}

\begin{table}[t!]
\centering
\caption{The change of $\bx_{[i:i+4]}$ and $\bpsi(\bx)_{[i:i+3]}$ after deleting $u_{i+1}u_{i+2}$}\label{tab:induced_deletion}
\begin{tabular}{|l||ll|}
\hline
$\bx_{[i:i+4]}\rightarrow \bx'_{[i:i+2]}$  & $  \bpsi(\bx)_{[i:i+3]}\rightarrow \bpsi(\bx')_{[i:i+1]}$ &   \\ \hline
$x_{i}101x_{i+4} \rightarrow x_{i}1x_{i+4}$    & $c_{i}11c_{i+3}\rightarrow c_{i}c_{i+3}$& 1) $c_{i}11c_{i+3}\rightarrow c_{i}c_{i+3}$           \\ \hline
\multirow{2}[0]{*}{$x_{i}100x_{i+4} \rightarrow x_{i}0x_{i+4}$} & \multirow{2}[0]{*}{$c_{i}10c_{i+3}\rightarrow {\bar{c}_{i}}c_{i+3}$} & 2) $010c_{i+3}\rightarrow 1c_{i+3}$ \\&       & 3) $110c_{i+3}\rightarrow 0c_{i+3}$ \\
\hline
\multirow{2}[0]{*}{$x_{i}001x_{i+4} \rightarrow x_{i}0x_{i+4}$} & \multirow{2}[0]{*}{$c_{i}01c_{i+3}\rightarrow c_{i}{\bar{c}_{i+3}}$} & 4) $c_{i}010\rightarrow c_{i}1$ \\&       & 5) $c_{i}011\rightarrow c_{i}0$      \\ \hline
\multirow{2}[0]{*}{$x_{i}100x_{i+4} \rightarrow x_{i}1x_{i+4}$} & \multirow{2}[0]{*}{$c_{i}10c_{i+3}\rightarrow c_{i}{\bar{c}_{i+3}}$} & 6) $c_{i}100\rightarrow c_{i}1$ \\&       & 7) $c_{i}101\rightarrow c_{i}0$\\ \hline
\multirow{2}[0]{*}{$x_{i}001x_{i+4} \rightarrow x_{i}1x_{i+4}$} & \multirow{2}[0]{*}{$c_{i}01c_{i+3}\rightarrow {\bar{c}_{i}}c_{i+3}$} & 8) $001c_{i+3}\rightarrow 1c_{i+3}$ \\&       & 9) $101c_{i+3}\rightarrow 0c_{i+3}$
\\ \hline
$x_{i}000x_{i+4} \rightarrow x_{i}0x_{i+4}$      & $c_{i}00c_{i+3}\rightarrow c_{i}c_{i+3}$& 10) $c_{i}00c_{i+3}\rightarrow c_{i}c_{i+3}$         \\ \hline
\end{tabular}
\end{table}

Using an averaging argument, we arrive at the following corollary.
\begin{corollary}
%For integers $n$, $0\le a < 2n$, $0\le b < q$ and $0\le c < q$, there exists a code $\cC(n)$ with redundancy $\log n+2\log q+1$.
There exists a code $\cC_I(a,b,n)$ of length $n$ capable of correcting an induced deletion of length $2$ with redundancy at most $\log n+2\log q+1$.
\end{corollary}

\begin{example}
Suppose $\bu=\left(1,0,6,7,6,2,3,5\right)\in\Sigma_8^8$ and the 4th symbol is deleted, so the retrieved sequence is $\bu'=\left(1,0,6,2,3,5\right)\in\Sigma_8^6$. The decoding process can be shown as the following:
\begin{enumerate}
    \item $\bu^o=\left(1,6,6,3\right)$, $\bu^e=\left(0,7,2,5\right)$, and $\bpsi(\bphi(\bu^o)\circ\bphi(\bu^e))=\left(0,0,0,1,0,0,1,1\right)$. Thus, $a=3$, $b=0$ and $c=6$.
    \item $\bu'^o=\left(1,6,3\right)$, $\bu'^e=\left(0,2,5\right)$. $\bpsi(\bphi(\bu'^o)\circ\bphi(\bu'^e))=\left(0,0,0,1,1,1\right)$ and $w=3$. We can get that $\Delta=4<2w$ and it is even. It belongs to Case 2 in the Proof of Theorem~\ref{thm:LC-alternative} and we insert two 0s on the left of the \textit{second} 1 from the right due to $R_1=2$. Thus, we can recover $\bpsi(\bphi(\bu^o)\circ\bphi(\bu^e))=\left(0,0,0,1,0,0,1,1\right)$.
    \item Due to $b=0$ and $c=6$, the value of deleted symbol in $\bu^o$ and $\bu^e$ are 6 and 7, respectively.
    \item From the second step, we can get $\bphi(\bu^o)=\left(1,1,0,0\right)$. Since $\bu'^o$, $\bphi(\bu^o)$ and the value of deleted symbol are known, $\bu^o=\left(1,6,6,3\right)$ can be recovered via the decoding process for Tenengol'ts code $\cC_T(a,b,n)$. Also, $\bu^e=\left(0,7,2,5\right)$ can be recovered through the same way. Therefore, we can recover $\bu=\left(1,0,6,7,6,2,3,5\right)$. 
\end{enumerate}
\end{example}

\iftrue

\section{Non-binary $2$-burst-error-correcting codes}\label{sec:noncor2}
In this section, we propose a non-binary code correcting a burst of at most 2 deletions with redundancy $\log n+O(\log q\log\log n)$, for even $q$. As discussed previously, we begin by describing a simple mapping that converts the problem at hand of correcting deletions of $q$-ary symbols from sequences of length $n$ to the problem of correcting deletions of symbols contained within a binary matrix comprised of $\lceil\log q\rceil$ rows and $n$ columns. Our approach will be to encode the first row of this matrix using a pattern length limited (PLL) Levenshtein code, described in Section~\ref{section:PLL}, which provides a range for the position of the error. The remaining rows of the matrix are encoded using a $P$-bounded Levenshtein code, presented in Subsection~\ref{section:Pbound}, which can correct a burst of at most 2 deletions given that the deletion is known to occur within a specific range of positions. The overall construction for correcting a burst of at most two deletions and its redundancy is presented in Subsection~\ref{section:finalcons}. 
\iffalse
We map non-binary sequences to binary matrices, given in the Subsection~\ref{section:mapping}, such that deletions occur in the same position in each row of the matrix. The rows are encoded with different codes, which together allow us to determine the position of the deletions. For the first row, we use a pattern length limited (PLL) Levenshtein code, described in Subsection~\ref{section:PLL}, which provides a range for the position of the error. %which is encoded as a Levenshtein code where we restrict any substring with the period 2 has length at most $\lceil \log n \rceil+5$. 
The following rows utilize the $P$-bounded Levenshtein code, presented in Subsection~\ref{section:Pbound}, which can correct a burst of at most 2 deletions given that the deletion occurred within $P$ consecutive known positions. The overall construction for correcting a burst of at most 2 deletions and its rate are presented in Subsection~\ref{section:finalcons}.
\fi

\subsection{Mapping from non-binary to binary}\label{section:mapping}

For a non-binary sequence $\bu \in \Sigma_q^n$, for even $q$, define the binary representation matrix $A(\bu)\in\Sigma_2^{\lceil \log q\rceil\times n}$ as
\begin{equation*}
    A(\bu) =\left[ \begin{matrix}
\bx_{1}\\\bx_{2}\\\vdots\\\bx_{\lceil \log q\rceil}
\end{matrix}\right] = \left[ \begin{matrix}
{x_{1,1}}&{x_{1,2}}& \cdots &{x_{1,n}}\\
{x_{2,1}}&{x_{2,2}}& \cdots &{x_{2,n}}\\
\vdots&\vdots&\vdots&\vdots\\
{x_{\lceil \log q\rceil,1}}&{x_{\lceil \log q\rceil,2}}& \cdots &{x_{\lceil \log q\rceil,n}}
\end{matrix} \right]
\end{equation*}
Let $A(\bu)_i$ denote the $i$th row of $A(\bu)$, where $A(\bu)_i=\bx_i$. The $j$th column of $A(\bu)$ is the binary representation of $u_j$, i.e., $u_j=\left[x_{1,j},x_{2,j},\dotsc,x_{\lceil \log q\rceil,j}\right]^T$, where $x_{1,j}$ is the least significant bit (LSB). Therefore, the non-binary sequence $\bu$ is converted to a  binary matrix with $\lceil \log q\rceil$ rows and $n$ columns. It is straightforward to observe that a burst of at most 2 deletions in $\bu$ corresponds to the deletion of at most 2 adjacent columns in  $A(\bu)$. %It is worth noting that the deletion positions are the same in each row, so we encode the first row and it will provide deletion position information for the remaining $\lceil \log q\rceil-1$ rows.

\subsection{Pattern length limited (PLL) Levenshtein Code}\label{section:PLL}

From the decoder of the Levenshtein code, we only can find the substring of period two from which the symbols were deleted, but not the exact deletion position. For example, suppose the underlined bits in $11\underline{01}01011$ are deleted and we receive $1101011$. While we can recover the original sequence, we only know the deletion occurred in the underlined substring $1\underline{1010101}1$, which has period 2. This creates difficulties when binary codes are used to correct errors in $q$-ary sequences.

If two adjacent bits are deleted, the deleted bits can be $11$, $10$, $01$, and $00$. If these are deleted from substrings of the form $1111\dotsm11$, $1010\dotsm10$, $0101\dotsm01$ and $0000\dotsm00$, respectively, the deletion position is ambiguous. The period of all of these patterns is 2. If only one bit is deleted, the deleted bit can be $1$ or $0$. The position of such deletion in the patterns $11\dotsm1$ and $00\dotsm0$, respectively, will be ambiguous. The period of these two patterns is also 2. (Note that, e.g., $00$ and $000$ have period 2.) Therefore, in order to narrow the range of possible deletion positions when a burst of length at most 2 occurs, we need to restrict the length of the patterns with period 2. 

It was proven in \cite{schoeny2017codes}, the redundancy of the code $\{ \bx\in \Sigma_2^n: L(\bx,1)\le \lceil \log n \rceil+1\}$ ($L(\bx,1)$ also denotes the length of the longest run in $\bx$) asymptotically approaches
$\log(e)/2 \approx 0.36$.
This idea is followed in \cite{chee2018coding} where it was shown that the redundancy of the code $\{ \bx\in \Sigma_2^n: L(\bx,2)\le \lceil \log n \rceil+2\}$ is at most 0.36 bits. Define the set of binary sequences of length $n$ where any substring with period 2 has length at most $r$ as $\pll(r,n)$. Here, we provide an explicit encoding algorithm to construct the code $\{ \bx: L(\bx,2)\le \lceil \log n \rceil+5\}$ with redundancy of 2 bits. This is shown in Algorithm~\ref{alg:L2}, %\textcolor{red}{Also can we use the notation $PLL(r,n)$ to describe the output of the algorithm rather than just ``encoded sequence''?}
where $b(i)$ represents the $\lceil \log n\rceil$-bit binary representation of $i$.

\begin{lemma}\label{lem:L2}
Given a sequence $\bx \in \Sigma_2^n$, Algorithm~\ref{alg:L2} outputs a unique sequence $\by\in \pll(\lceil \log n \rceil+5,n+2)$ where any substring of $\by$ with period 2 has length at most $\lceil \log n \rceil+5$. The redundancy of the encoding is 2 bits. 
\end{lemma}

\begin{proof}
Note that initially, $|\by|=n+2$, and every time a substring of length $\lceil\log n\rceil+5$ is deleted, a string of length $1+2+\lceil\log n\rceil+2$ is appended to $\by$. Hence, $|\by|$ remains $n+2$ during the encoding process.

\begin{algorithm}[t]
\label{alg:L2}
\SetAlgoLined
\KwInput{Sequence $\bx\in \Sigma_2^n$}
\KwOutput{Encoded sequence $\by\in \pll(\lceil \log n \rceil+5,n+2)$}
$\by\leftarrow(x_1,x_2,\dotsc,x_n,1,0),\ i\leftarrow 1,\ n'\leftarrow n$.

\While{$i\le n'-\lceil \log n \rceil-3$}{
\eIf{$L(\by_{[i,i+\lceil \log n \rceil+5]},2)=\lceil \log n \rceil+6$}
{Delete $\by_{[i,i+\lceil \log n \rceil+4]}$ from $\by$ and append $(0, \by_{[i,i+1]},b(i),11)$ to the end of $\by$.\\
$n'\leftarrow n'-\lceil \log n \rceil-5,\ i\leftarrow 1$}
{
   $i\leftarrow i+1$
  }
}
 \caption{Pattern Length Limited Encoding}
\end{algorithm}

Let $\bs_i=\by_{[i,i+\lceil\log n\rceil+5]}$. We show that when the algorithm terminates, no $\bs_i$ has period 2:
\begin{itemize}
    \item For $i\le n'-\lceil \log n \rceil-3$, after the termination of the algorithm, any $\bs_i$ with period 2 has been deleted. 
    \item If $n'-\lceil \log n \rceil-2\le i\le n'+1$, and $n'\neq n$, then $\bs_i$ contains 100 and thus its period is not 2. If $n'=n$, no such $\bs_i$ exists as $(n'-\lceil\log n\rceil -2) +(\lceil\log n\rceil+5)>n+2$.
    \item If $i>n'+1$, then each $\bs_i$ will contain both $0$ and $11$ as these are parts of each appended string. A string containing both $0$ and $11$ cannot have period 2. 
\end{itemize}
We next show that the encoding algorithm is invertible by showing that each deletion of a substring of length $\lceil \log n\rceil + 5$ is invertible. After such a deletion, $\by$ ends with a substring of the form $0ab\bh11$, where $a,b\in \Sigma_2, \bh\in \Sigma_2^{\lceil\log n\rceil}$. We can invert this deletion by inserting a string of form $abab\dotsm$ of length $\lceil \log n\rceil + 5$ in $\by$ at the position with binary representation given by $\bh$ and removing $0ab\bh11$. So given $\by$, by inverting each deletion, starting from the last one, we can recover $\bx$. It follows that the encoding is injective and the redundancy is 2 bits.
\end{proof}

\begin{example}
Suppose $\bx=(1101010101010101)\in \Sigma_2^{16}$, where $L(\bx,2)=15$. The encoding process of Algorithm~\ref{alg:L2} is as follows:
\begin{enumerate}
\item Initially, $\by=(1101010101010101\mathbf{10})$, $i=1$, $n'=16$.
\item When $i=2$, we have $L(\by_{[2,11]},2)=10=\log n+6$. We delete $\by_{[2,10]}$ and append $(010001011)$ to the end of $\by$. So $\by$ becomes $(1010101\mathbf{10}\bcomment{010001011})$. We then let $i=1$ and $n'=7$.
\item Then $n'-\log n-3=0$ and the encoding process ends.
\end{enumerate}
It can be seen that $L(\by,2)=7\le \lceil \log n \rceil+5$.
\end{example}

We can then construct pattern length limited Levenshtein codes in which the ambiguity of the position of the deletion is restricted to an interval of length $\lceil\log n\rceil+5$.

%\begin{construction}\label{const:pll}
%For arbitrary integers $n$ and $0\le a < 2n$, define the pattern length limited Levenshtein code $\plld\cC_L(a,\lceil \log n\rceil +5,n)$ as %the intersection of the code $\cB_a(n)$ and $\cS_n(\lceil \log n\rceil +5)$:
%\begin{equation*}
%    \plld\cC_L(a,\lceil \log n\rceil +5,n)=\cC_L(a,n)\cap \pll(\lceil \log n\rceil +5,n).
%\end{equation*}
%\end{construction}
\begin{lemma}
For integers $n$ and $a\in[[2n]]$, the pattern length limited Levenshtein code
\begin{equation*}
    \plld\cC_L(a,\lceil \log n\rceil +5,n)=\cC_L(a,n)\cap \pll(\lceil \log n\rceil +5,n).
\end{equation*}
has redundancy at most $\log n+3$.
\end{lemma}
\begin{proof}
By Lemma~\ref{lem:L2}, $|\pll(\lceil \log n\rceil +5,n)|\ge |\pll(\lceil \log (n-2)\rceil +5,n)|\ge 2^{n-2}$. Hence, there exists some value of $a$ such that $|\plld\cC_L(a,\lceil \log n\rceil +5,n)|\ge 2^{n-2}/(2n)$.
\end{proof}

Note that at this point, we can use the code $\plld \cC_L$ to correct the burst of length at most 2 which occurs in the first row of matrix $A(u)$. In addition, due to the fact that we have limited the length of substrings of period two in our codewords, it follows that we can determine the range of where the burst of deletions has occurred in the remaining rows to within $\log n + 5$ positions. The code described in the next subsection leverages this information to correct the burst of deletions of length at most 2 in the remaining rows. We note that a similar construction appears in \cite{schoeny2017codes} for a code that can correct a \textit{single} deletion given its approximate location. The key difference between their result and ours is that the code described in the next subsection can correct a \textit{burst} of length \textit{at most two} deletions given its approximate location.

\subsection{P-bounded codes for correcting a burst of at most 2 deletions}\label{section:Pbound}

Next we show that the $P$-bounded Levenshtein code defined below can make use of the information obtained from the first row of $A(u)$ to correct the deletions in the remaining rows.

%\begin{construction}
%For $c\in[[2P]]$ and $d\in[[3]]$, define the %$P$-bounded Levenshtein code $\cC_L(c,d,P,n)$ as
%\begin{equation*}
%    \cC_L(c,d,P,n) = \Large\{\bx\in\Sigma_2^n: \VT(\bpsi(\bx))\equiv c \mod 2P,\quad \wt(\bphi(\bx)) = d \mod 3\Large\}
%\end{equation*}
%\end{construction}

\begin{theorem}
For all $c\in[[2P]]$ and $d\in[[3]]$, the $P$-bounded Levenshtein code  
\begin{equation*}
    \cC_L(c,d,P,n) = \left\{\bx\in\Sigma_2^n: \VT(\bpsi(\bx))\equiv c \mod 2P,\quad \wt(\bpsi(\bx)) = d \mod 3\right\}
\end{equation*}
can correct a burst of at most 2 deletions with redundancy $\log P+\log 6$, if the position $i$ of the first deleted symbol in $\bx$ is known to lie within an interval of length $P$. More precisely, $i$ can be uniquely determined provided $m$ is known where $i\in [m,m+P-1]$.
\end{theorem}

\begin{proof}
For a codeword $\bx$, let $\bx'$ be obtained from $\bx$ after a deletion of a single symbol or two adjacent symbols. Also, let $\by = \bpsi(\bx)$ and $\by'= \bpsi(\bx')$. Further, let $\Delta = \VT(\by)-\VT(\by')$ and $\Delta_w = \wt(\by)-\wt(\by')$. We can determine the number of deleted bits based on the length of $\bx'$. We show that from $\by'$ we can recover $\by$, and then recover $\bx$, since $\bpsi$ is invertible. 

Let either $x_i$ or $x_ix_{i+1}$ be deleted from $\bx$. For uniqueness, in the former case we assume $x_i\neq x_{i+1}$ and in the latter $x_i\neq x_{i+2}$. Similar to the proof of Theorem~\ref{thm:LC-alternative}, the error in $\bx$ can affect $\by$ in the following ways:
\begin{itemize}
    \item If the first one or two symbols of $\bx$ are deleted, then the first one or two symbols of $\by$ are deleted, respectively.
    \item If $x_i, i \in \{2,3,\ldots,n\}$ is deleted, then $y_{i-1}y_{i}$ is replaced by $y_{i-1}\oplus y_{i}$. The possible error patterns are: $01\to 1, 10\to 1, 11\to 0$. That is, a 0 is deleted from position $i-1$ or $i$, or an 11 in position $i-1$ is replaced with 0.
    \item If $x_{i}x_{i+1}, i \in \{2,3,\ldots,n-1\}$ are deleted, then $y_{i-1}y_{i}y_{i+1}$ is replaced by $y_{i-1}\oplus y_{i}\oplus y_{i+1}$. The error patterns are: $000\to 0, 001\to 1, 010\to 1, 011\to 0,100\to 1, 101\to 0, 110\to 0, 111\to 1$. That is a 00, an 11, two 1s around a 0, or two 0s around a 1 are deleted. 
\end{itemize}
From this point on, we will consider the errors in $\by$ only. Let $p$ denote the position of the first deleted or replaced bit in $\by$ (where $\by = \bpsi(\bx)$ and $\by'= \bpsi(\bx')$). Furthermore, let $R_1, L_1$ be the number of 1s on the right and the left of the altered bits, respectively. Similarly, define $R_0, L_0$. Let $w$ denote the weight of $\by'$. For example, for $\by = 010\boldsymbol{1}\boldsymbol{0}\boldsymbol{1}00\to \by = 010\boldsymbol{0}00$, we have $p=4, L_0=2, L_1 = 1, R_0 = 2, R_1 = 0, w = 1$. From the previous discussion, we have $p\in[\max(m-1,1),m+P-1]$.

First, we consider $|\by'|=|\by|-1$. We find below how $\Delta$ changes. The cases can be distinguished based on the change in $\Delta_w$. 
\begin{enumerate}
    \item If a 0 is deleted, we have $\Delta_w=0$. Then $\Delta=R_1$ and $p\in\{i-1,i\}$ where $\{i-1,i\} \in[\max(m-1,1),m+P-1]$. Let $r_1$ be the number of 1s on the right of position $m+P-1$ in $\by$, exclusive. Then, $r_1\le R_1\le r_1+P$. Note that $r_1$ equals the number of 1s on the right of the position $m+P-2$ in $\by'$ and thus it is known. Let another true value of $R_1$ satisfy the equation $\Delta\equiv R_1'(\bmod 2P)$. Hence, $R_1\equiv R_1'$. But this is impossible as $0\le R_1-R_1'\le P$ since $R_1\in[r_1,r_1+P]$.
    \item If a 1 is deleted, we have $\Delta_w=1$. This can only be the first element of $\by$.
    \item If an 11 is replaced with a 0, we have $\Delta_w=2$. Then $\Delta = 2p+1+R_1$ and $p\in[m-1,m+P-2]$. Suppose on the contrary that $p' > p$ denotes the location of the deletion where $p'$ satisfies the equation $\Delta\equiv 2p'+1+R_1'(\bmod 2P)$. Hence, $2(p'-p)\equiv R_1-R_1'$. But this is impossible as $0\le R_1-R_1'\le p'-p\le P-1$.
\end{enumerate}

Next, we consider $|\by'|=|\by|-2$: The following cases can be distinguished based on the change in $\Delta_w$ and the parity of $\Delta \bmod 2P$.
\begin{enumerate}
    \item If $010\to 1$ occurs, we have $\Delta_w=0$. Then $\Delta = 2R_1+1$ and $p\in[m-1,m+P-3]$. Let another true value of $R_1$ satisfy the equation $\Delta\equiv 2R_1'+1(\bmod 2P)$. Hence, $R_1\equiv R_1'$. But this is impossible as $0\le R_1-R_1'\le P-2$.
    \item If 00 is deleted, we have $\Delta_w=0$. Then $\Delta = 2R_1$ and $p\in\{i-1\}\in[m-1,m+P-3]$ or $p\in\{i\}\in[m,m+P-2]$. Note that in this case $\Delta$ is even whereas in Case 1, $\Delta$ was odd. Also, $\Delta \bmod 2P$ has the same parity as $\Delta$, which helps us to distinguish Case 1 and Case 2. The proof of this item is the same as in Case 1. 
    \item If $10$ or $01$ is deleted from the beginning, we have $\Delta_w=1$. Then $\Delta=2R_1+1=2w+1$ or $\Delta = 2w+2$, respectively. When $\Delta_w=1$, prepend $10$ or $01$ depending on the parity of $\Delta$. The same as in previous cases, $\Delta \bmod 2P$ has the same parity as $\Delta$.  
    \item If 11 is deleted, we have $\Delta_w=2$. Then $\Delta = 2p+1+2R_1$ and $p\in\{i-1\}\in[m-1,m+P-3]$ or $p\in\{i\}\in[m,m+P-2]$. Suppose, on the contrary there exists a $p'>p$ that satisfies the equation $\Delta\equiv 2p'+1+2R_1'(\bmod 2P)$. Hence, $p'-p\equiv R_1-R_1'$. But since $0\le R_1-R_1'\le p'-p\le P-2$, the only possible ambiguous pattern is that the symbols between $p$ and $p'$ are all 1. Therefore, the two sequences that result from inserting $11$ in position $p$ or in position $p'$ are equivalent. %Further, since the deletion in this case in $11$, hence $p$ and $p'$ in this ambiguous pattern are equivalent.
    \item If $101\to 0$ occurs, we have $\Delta_w=2$. Then $\Delta = 2p+2+2R_1$ and $p\in\{i-1\}\in[m-1,m+P-3]$. Note that in this case $\Delta$ is even whereas in the previous case it was odd and $\Delta \bmod 2P$ has the same parity as $\Delta$. The same as Case 4, the only possible ambiguous pattern is that the symbols between $p$ and $p'$ are all 1. Further, since the error pattern in this case is $101\to 0$, hence this ambiguous pattern is impossible. 
\end{enumerate}

The redundancy follows from the fact that there are $6P$ options for choosing $c,d$.
\iffalse
In all cases, $\Delta\le 2(n-2)+3=2n-1$, can be achieved in case 5. The case can be distinguished based on the based on the comparison of $\Delta$ with $w,w+3$ and its parity. To recover $\by$: In case 1, we insert two 0s on both sides of the $(R+1)$th 1 from the right. In case 2, insert two 0s on the left of the $R_1$th 1. In case 3, 10 or 01 must be prepended, depending on the parity of $\Delta$. In case 4, noting that $\Delta=2w+2L_0+3$, we insert 11 immediately after the $L_0$th 0. Finally, in case 5, we insert two 1s on the sides of the $(L_0+1)$th 0, where $L_0 = \Delta/2-w-2$.
\fi
\end{proof}

As mentioned before, a similar construction shows in \cite{schoeny2017codes} for a code that can correct a single deletion given its approximate location. The difference of construction between our $P$-bounded code with their result is that the parity constraint is now $\bmod 3$ rather than $\bmod 2$. 

The proof of our construction follows from 
Theorem~\ref{thm:LC-alternative} via the case-by-case discussion. The codes in \cite{schoeny2017codes} identify the change in the syndrome in a substring of bounded length. We also present the proof for our $P$-bounded code by checking the change in the syndrome based on runs in a substring of bounded length in \cite{wang2021non}.

\subsection{Construction for correcting a burst of at most 2 deletions}\label{section:finalcons}

\begin{theorem}
Let $q$ be an even integer. For all $a\in[[2n]]$, $c_i\in[[2(\lceil \log n\rceil+5)]]$ and $d_i\in[[3]]$, %the code %$\cC(n)$ 
\begin{align*}
    \cC_{2B}(a,c_i, d_i,P,n) =\Big \{&\bu\in \Sigma_q^n: A(\bu)_1\in \plld\cC_L(a,n,\lceil \log n\rceil+5),\quad \\
    &A(\bu)_i\in \cC_L(c_i,d_i,\lceil \log n\rceil+5,n), \forall i \in \left [2,\lceil\log q\rceil \right] \Big\}
\end{align*}
is a $2$-burst-error-correcting code. Furthermore, there exists choices for $a,c_i,d_i$ such that the redundancy is at most $\log n + \log q(\log(\lceil\log n\rceil+5)+\log 6)+3$.
%$\log n + (\lceil \log q\rceil-1)(\log(\lceil\log n\rceil+6)+\log 6)+3$.
\end{theorem}

\iftrue
\begin{proof}
If a burst of at most 2 deletions occurs in $\bu\in\cC_{2B}(a,c_i,d_i,P,n)$, at most 2 adjacent columns are deleted in the binary matrix $A(\bu)$. 
%The first row $A(\bu)_1$ of $A(\bu)$ belongs to a $\pll\cB_a(n,\lceil \log n\rceil+5)$ code. 
Since the LSB of the binary representation of $u\in\Sigma_q$ can take any value regardless of the other bits, the PLL constraint can be applied to $A(\bu)_1$. The decoder for $\plld\cC_L(a,\lceil \log n\rceil+5,n)$ can insert the deleted bits into $A(\bu)_1$. Since any substring with period 2 in $A(\bu)_1$ has length at most $\lceil \log n\rceil +5$, we can find an interval of length at most $\lceil \log n\rceil +5$ in which the deletion has occurred. Then, $\cC_L(c_i,d_i,\lceil \log n\rceil+5,n)$ can correct the deletions in the remaining $\lceil\log q\rceil-1$ rows of $A(\bu)$. %The redundancy follows from the number of choices for $a,c_i,d_i$.
\end{proof}
\fi
The code has redundancy $\log n+O(\log q\log\log n)$ for even $q$. We note that a lower bound on the redundancy of the $q$-ary code correcting a burst of 2 deletions, which is also a lower bound for the redundancy of correcting a burst of \textit{at most} 2 deletions, is $\log n+O(\log q)$ in Theorem~\ref{thm:nonasymptotic}.

\iffalse
\begin{corollary}
For integers $n$, $0\le a<2n$, $0\le c<2(\lceil \log n\rceil+6)$ and $0\le d<3$, there exists a code $\cC_{a,c,d}(n)$ with redundancy at most $\log n + (\lceil \log q\rceil-1)(\log(\lceil\log n\rceil+6)+\log 6)+3$.
\end{corollary}
\fi

{\textit{Remark:}} To get the systematic construction of the non-binary code for correcting a burst of at most 2 deletions, the encoding and decoding process can be outlined as the following:
\begin{itemize}
    \item {\textbf{Encoding:}} Suppose the syndrome of Levenshtein code $C_L(a,n)$ is $H_{cor,2}(\bx)$. The $i$th row of the binary representation matrix $A(\bu)$ can be expressed as $\left(A(\bu)_i,0,0,1,H_{cor,2}(A(\bu)_i)\right)$ of length $N$.
    \item {\textbf{Decoding:}} Denote the length of received sequence $\bu'$ as $N'$ and let $m=n+3-(N-N')$. First, we need to determine the value of $A(\bu')_{[1,m]}$, where $A(\bu')_{[1,m]}$ denote the element in the 1st row and $m$th column of the matrix $A(\bu')$:
    \begin{enumerate}
    \item If $A(\bu')_{[1,m]}=1$, this means the deletion occurs in information bits because the 1 in the protecting bits has shifted $N-N'$ positions to the left. Then, $A(\bu)_i$ can be recovered by their corresponding syndromes. 
    \item If $A(\bu')_{[1,m]}=0$, this means the deletion occurs in the syndrome part because the 1 in the protecting bits does not shift. Then, we can directly get $\bu$.
\end{enumerate}
\end{itemize}
The total redundancy of this systematic construction is $\log q\log n+O(\log q)$, which is much larger than our construction proposed in this section. As will be discussed in Section~\ref{sec:conclusion}, the design of a systematic encoding scheme for a low-redundancy non-binary $2$-burst-error-correcting code remains a direction for future work.

\section{Non-binary $t$-burst-error-correcting codes}\label{sec:noncort}

In this section, we construct non-binary $t$-burst-error-correcting codes. We will employ the same mapping from non-binary to binary symbols that was used in the previous section whereby we will interpret our length $n$ codewords as $\lceil\log q\rceil \times n$ binary matrices. In this setting, rather than setting the first row of this matrix to be a binary code capable of correcting a burst of $2$ deletions, we will instead set the first row of our codeword matrix to be a code capable of correcting a burst of $t$ deletions. To this end, we will leverage the construction from \cite{lenz2020optimal}, as described in detail in Section~\ref{sec:noncort_loc}. 

Suppose that a burst of length at most $t$ occurs to one of our codewords under the assumption that the first row of each of our codewords belongs to a code capable of correcting a burst of at most $t$ deletions. If we also require that the first row of our codeword matrix is a so-called $(\bw,\delta)$-dense string, then it can be shown that we can approximately determine the location of the burst of deletions (to within roughly $O(\log n)$ positions). It then remains to correct the deletions in each of the remaining rows using this knowledge. A straightforward approach is to introduce a set of roughly $\log n$ parity constraints, which implies the redundancy of this naive approach is much larger than $\log n$. A better approach, which requires only $\log n + O(\log q \log \log n)$ bits of redundancy, is to first partition each row of our codeword matrix into blocks of size $O(\log n)$ and then for each row to enforce two parity constraints on the syndrome of each of these blocks. This approach requires only at most $O(\log \log n)$ additional bits of redundancy for each row of the codeword matrix, which is significantly less than the naive approach. This portion of the construction is explained in more detail in Section~\ref{sec:noncort_p}.

\subsection{Locating the burst of deletions}\label{sec:noncort_loc}

%\rcomment{We should explicitly mention hear (again) that we use use the same construction as [7] and (more importantly) we introduce an explicit encoding method which is not given in [7].}

In this subsection, we utilize the code in Construction 1 from \cite{lenz2020optimal} to approximately determine the location of a burst of at most $t$ deletions (to within roughly $O(\log n)$ positions). Further, we provide an explicit encoding method which is not given in \cite{lenz2020optimal}, for the key step of this construction: encoding the binary sequence into a so-called $(\bw,\delta)$-dense string. 

Denote the indicator vector of the pattern $\bw$ as $\mathbbm{1}_{\bw}$:
\begin{equation}
    \mathbbm{1}_{\bw}(\bx)_i=\begin{cases}
    1, &{\text{if}}\; \bx_{[i,i+|{\bw}|-1]}=\bw \\
    0, &{\text{otherwise}}
    \end{cases}
\end{equation}
Let $n_{\bw}$ be the number of ones in $\mathbbm{1}_{\bw}$ and define $\balpha_{\bw}(\bx)$ to be a vector of length $n_{\bw}+1$ whose $i$-th entry is the distance between positions of the $i$-th and $(i+1)$-th 1 in the string $(1,\mathbbm{1}_{\bw},1)$.

Fix the pattern $\bw=\mathbf{0}^t\mathbf{1}^t$ and $\delta=t2^{2t+1}\lceil\log n\rceil$. We now introduce the set of $(\bw, \delta)$-dense strings:
\begin{align*}
    \cD_{\bw, \delta}(n) = \left\{ \bx \in \Sigma_2^n : \balpha_{\bw}(\bx)_i \leq \delta, \forall i \in [\left| \balpha_{\bw}(\bx) \right|]   \right \}.
\end{align*}
Since the length of $\bw$ is $2t$ and its occurrences are non-overlapping,  every component in $\balpha_{\bw}(\bx)$ has value at least $2t$. Furthermore, the number $n_\bw$ of patterns in $\bx$ is at most $n_{\bw}\leq\frac{|\bx|}{2t}$.

\begin{lemma} (\cite[Lemma 1]{lenz2020optimal})
For any $n\ge 5$, the number of $(\bw,\delta)$-dense strings of length $n$ is at least
\begin{equation*}
    |\cD_{\bw, \delta}(n)|\ge 2^n(1-n^{1-\log e})\ge 2^{n-1}
\end{equation*}
\end{lemma}
%\rcomment{I think this is the same as the result in [7]. In that case, using c.f. is not appropriate since it means "compare" rather than providing a reference for the lemma.}

\begin{lemma}\label{lem:locdels_qary}(\cite[Construction 1]{lenz2020optimal})\label{lem:loc_del}
For any integers $c_0\in[[4]]$ and $c_1<[[2n]]$, let
\begin{equation*}
    \cC_{loc}(c_0,c_1,n) = \{ \bx\in \cD_{\bw, \delta}(n): \quad
    n_{\bw}(\bx) \equiv c_0 \pmod{4}, \quad
    \VT(\balpha_{\bw}(\bx))  \equiv c_1 \pmod {2n}\}.
\end{equation*}
The code $\cC_{loc}(c_0,c_1,n)$ with redundancy $\log n+4$ is capable of locating the burst of deletions to an interval of length at most $\delta=t2^{2t+1}\lceil\log n\rceil$.
\end{lemma}

Next, we will provide an explicit construction for encoding binary sequences into $(\bw,\delta)$-dense strings. The key idea of the encoding is to identify a substring of length $\delta$ in the input string $\bx$ that does not include a pattern $\bw$, and each such substring will be removed from $\bx$. Then, a compressed version of the removed substring is appended to the end of the input string along with a pattern. We first demonstrate that all substrings of length $\delta$ without a pattern can be compressed.

\begin{proposition}
Let $\cS$ be the set of strings of length $\delta$ that do not contain a pattern $\bw=\mathbf{0}^t\mathbf{1}^t$. Then, every string $\bs\in\cS$ can be compressed into a string $g(\bs)\in \Sigma_2^{\delta-\lceil\log n\rceil-4t-2}$, where function $g$ is an invertible map such that $g$ and $g^{-1}$ can be computed in $O(\delta)$ time. 
\end{proposition}

The compression works as follows. Split the string $\bs$ into $2^{2t}\lceil\log n\rceil$ substrings with each length of $2t$. Each substring can be represented by a symbol from the alphabet of size $2^{2t}-1$ since no substring can be equal to $\mathbf{0}^t\mathbf{1}^t$. In other words, the string $\bs$ can be represented by a string $\bv$ consisting $2^{2t}\lceil\log n\rceil$ symbols from the the alphabet of size $2^{2t}-1$. The number of bits $n_{\bv}$ required to represent $\bv$ is
\begin{align*}
    n_{\bv} &=\lceil\log \left(2^{2t}-1\right)^{2^{2t}\lceil\log n\rceil}\rceil \\
    &=\lceil\log \left(1-2^{-2t}\right)^{2^{2t}\lceil\log n\rceil}\rceil+\left((2t) 2^{2t}\lceil\log n\rceil\right) \\
    & \stackrel{(a)}{\leq} \lceil\log n\rceil \log \left(\frac{1}{e}\right)+1+\delta \\
    & \leq \delta-1.4\lceil\log n\rceil+1 \\
    & \stackrel{(b)}{\leq} \delta-\lceil\log n\rceil-4t-2
\end{align*}
where (a) follows from the fact that for the function $(1-1/x)^x$ is increasing in $x$ for $x>1$ and $\lim_{x\rightarrow\infty}(1-1/x)^x=1/e$, and (b) follows from $0.4\lceil\log n\rceil\ge 4t+3$ for large value $n$. This compression 
contains a conversion from binary to decimal, so the total complexity of function $g$ is $O(\delta)$.

\begin{algorithm}[t]

\label{alg:pattern_dense}
\SetAlgoLined
\KwInput{Sequence $\bx\in \Sigma_2^n$}
\KwOutput{Encoded sequence $E_{\bw, \delta}(\bx)$}

\textbf{Initialization:} Let $E_{\bw, \delta}(\bx)=\bx$ and $n'=n$. Append $\left(\mathbf{0}^t\mathbf{1}^t\mathbf{0}^t\mathbf{1}^t\right)$ at the end of the sequence $E_{\bw, \delta}(\bx)$.

\textbf{Step 1:} If there exists an integer $i\in[1,n']$ such that for every $j\in[i,i+\delta-2t]$ it holds that $\left(E_{\bw, \delta}(\bx)_{j},E_{\bw, \delta}(\bx)_{j+1},\dotsc,E_{\bw, \delta}(\bx)_{j+2t-1}\right)\neq \mathbf{0}^t\mathbf{1}^t$, go to Step 2 or 3. Else go to Step 4.
    \begin{enumerate}
        \item \textbf{Step 2:} If $i\leq n'-\delta+1$, then we delete $\left(E_{\bw, \delta}(\bx)_{i},E_{\bw, \delta}(\bx)_{i+1},\dotsc,E_{\bw, \delta}(\bx)_{i+\delta-1}\right)$ from $E_{\bw, \delta}(\bx)$ and append\\ $\left(i,g\left(E_{\bw, \delta}(\bx)_{i},E_{\bw, \delta}(\bx)_{i+1},\dotsc,E_{\bw, \delta}(\bx)_{i+\delta-1}\right),1,\mathbf{0}^t\mathbf{1}^t\mathbf{0}^t\mathbf{1}^t,0\right)$, where $i$ is the binary representation of index with length $\lceil\log n\rceil$. Let $n'=n'-\delta$ and $i=1$.
        \item \textbf{Step 3:} If $i> n'-\delta+1$, then we delete $\left(E_{\bw, \delta}(\bx)_{i},E_{\bw, \delta}(\bx)_{i+1},\dotsc,E_{\bw, \delta}(\bx)_{n'}\right)$ from $E_{\bw, \delta}(\bx)$ and append\\ $\left(i,g\left(E_{\bw, \delta}(\bx)_{i},E_{\bw, \delta}(\bx)_{i+1},\dotsc,E_{\bw, \delta}(\bx)_{n'},\mathbf{0}^{i+\delta-n'-1}\right),1,\mathbf{0}^t\mathbf{1}^t\mathbf{0}^{t-(i+\delta-n'-1)/2}\mathbf{1}^{t-(i+\delta-n'-1)/2},0\right)$, where $i$ is the binary representation of index with length $\lceil\log n\rceil$. Let $n'=n'-i$.
    \end{enumerate}
    
    \textbf{Step 4:} Output $E_{\bw, \delta}(\bx)$.
 \caption{Pattern-dense Strings Encoding}

\end{algorithm}

\begin{lemma}\label{lem:encodingdense}
Given a sequence $\bx\in\Sigma_2^n$, Algorithm~\ref{alg:pattern_dense} outputs a unique pattern-dense string $E_{\bw, \delta}(\bx)\in \cD_{\bw,\delta}(n+4t)$.
\end{lemma}

Notice that the length of $E_{\bw, \delta}(\bx)$ keeps stable during the encoding process. Each time we delete $\delta$ or $n'-i+1$ bits, we will append a new string with identical deleted length at the end of $E_{\bw, \delta}(\bx)$. Next, we want to show that the encoding sequence $E_{\bw, \delta}(\bx)$ satisfies the condition that each interval of length $\delta$ in $E_{\bw, \delta}(\bx)$ contains at least one pattern $\bw=\mathbf{0}^t\mathbf{1}^t$. It can be seen that for any $i\in \left[1,n'\right]$, if there exists $j\in[i,i+\delta-2t]$ such that $\left(E_{\bw, \delta}(\bx)_{j},E_{\bw, \delta}(\bx)_{j+1},\dotsc,E_{\bw, \delta}(\bx)_{j+2t-1}\right)\neq \mathbf{0}^t\mathbf{1}^t$, the corresponding substring beginning with $E_{\bw, \delta}(\bx)_i$ will be deleted in Step 2 or 3. And all of the new appended string contains at least a pattern $\bw$. The decoding process of recovering $\bx$ from $E_{\bw, \delta}(\bx)$ is given in Algorithm~\ref{alg:pat_dense_decoding}.

\begin{algorithm}[t]
\label{alg:pat_dense_decoding}
\SetAlgoLined
\KwInput{Sequence $E_{\bw, \delta}(\bx)\in \cD_{\bw,\delta}(n+4t)$}
\KwOutput{Decoded sequence $\bx\in\Sigma_2^n$}

\textbf{Initialization:} Let $\bx=E_{\bw, \delta}(\bx)$.

\If{ $x_{n+4t}=0$} 
{Find the beginning index $k$ of the second $\mathbf{0}$ run with length $t$ from the right to the left and let $l=n+4t-k$. Then, let $i$ be the decimal representation of $\left(x_{n+8t-l-\delta+1},x_{n+8t-l-\delta+2},\dotsc,x_{n+8t-l-\delta+\lceil\log n\rceil}\right)$. Let $\by$ be the sequence obtained by $g^{-1}\left(x_{n+8t-l-\delta+\lceil\log n\rceil+1}, x_{n+8t-l-\delta+\lceil\log n\rceil+2},\dotsc,x_{n+4t-l-2}\right)$. Next, delete $\left(x_{n+8t-l-\delta+1},x_{n+8t-l-\delta+2},\dotsc,x_{n+4t} \right)$ from $\bx$ and insert $\left(y_1,y_2,\dotsc,y_{\delta-4t+l}\right)$ at the location $i$ of $\bx$. Repeat.}

\If{$x_{n+4t}=1$} 
{Delete $\left(x_{n+1},x_{n+2},\dotsc,x_{n+4t}\right)$ and output $\bx$.}

\caption{Pattern-dense Strings Decoding}
\end{algorithm}

\iffalse
\textbf{Decoding:} The decoding process of recovering $\bx$ from $E_{\bw, \delta}(\bx)$ is given as the following:
\begin{enumerate}
    \item \textbf{Initialization:} Let $\bx=E_{\bp, \delta}(\bx)$.
    \item If $x_{n+4t}=0$, find the beginning index $k$ of the second $\mathbf{0}$ run with length $t$ from the right to the left and let $l=n+4t-k$. Then, let $i$ be the decimal representation of $\left(x_{n+8t-l-\delta+1},x_{n+8t-l-\delta+2},\dotsc,x_{n+8t-l-\delta+\lceil\log n\rceil}\right)$. Let $\by$ be the sequence obtained by $g^{-1}\left(x_{n+8t-l-\delta+\lceil\log n\rceil+1}, x_{n+8t-l-\delta+\lceil\log n\rceil+2},\dotsc,x_{n+4t-l-2}\right)$. Next, delete $\left(x_{n+8t-l-\delta+1},x_{n+8t-l-\delta+2},\dotsc,x_{n+4t} \right)$ from $\bx$ and insert $\left(y_1,y_2,\dotsc,y_{\delta-4t+l}\right)$ at the location $i$ of $\bx$. Repeat.
    \item If $x_{n+4t}=1$, delete $\left(x_{n+1},x_{n+2},\dotsc,x_{n+4t}\right)$ and output $\bx$.
\end{enumerate}
\fi

\subsection{P-bounded code for correcting a burst of at most $t$ deletions}\label{sec:noncort_p}

%\textcolor{red}{Remind the reader how is it that we use Lemma~\ref{lem:locdels_qary}. If we enforce that our codeword sequences merely belong to the code $\cC_{loc}(n,c_0,c_1)$, then this doesn't make any sense. }
%$From Lemma~\ref{lem:locdels_qary}, we can locate the deletion into a $P=O(\log n)$ interval, 

Next, we discuss how to correct a burst of at most $t$ deletions provided we know approximately where the deletions occur. We will make use of a code that was designed in prior work.
%\textcolor{red}{Remind the reader what the notation below means. For example, $f_t(\bx)\bmod a$ is a number $\{0,1,\ldots, a-1\}$, but you write that it's binary. It's ok to use an abuse of notation but you should notify the reader of it ahead of time. }

\begin{lemma}\label{lem:t_burst_base}
(c.f. \cite{sima2020syndrome}) For any $\bx\in\Sigma_2^k$, there exists an integer $a\leq 2^{2\log k+o(\log k)}$ and the labeling function $f_t: \Sigma_2^{k}\rightarrow\Sigma_{2^{\cR(t,k)}}$ where $\cR(t,k)\leq O(t^2\log(k+t))$ the systematic code %$\cC^{b_c}(n,t)$ with $n=k+4\log k+o(\log k)$,
\begin{equation}
    \cC_{{SB}}(n,t)=\left\{\by=\left(\bx,1,\mathbf{0}^t,\mathbf{1}^t,0,a,f_t(\bx)\bmod a\right): \bx\in\Sigma_2^k\right\}.
\end{equation}
is capable of correcting a burst of at most $t$ consecutive deletions, where $a$ and $\left(f_t(\bx)\bmod a\right)\in[[a]]$ are represented as binary vectors.
\end{lemma} 
Let $E_{SB} : \Sigma_2^{k} \to \Sigma_2^{4\log k + o(\log k)}$ denote the systematic encoder for $\cC_{SB}$ that for an input $\bx \in \Sigma_2^k$ outputs the information $(a, f_t(\bx)\bmod a)$ in binary.

Our next step is to introduce an additional constraint, which leverages the encoder from $C_{SB}(n,t)$, that allows us to recover the burst of at most $t$ deletions provided we know approximately where the burst of deletions occurs. The idea behind the approach is to first compute the non-systematic portion of the code $C_{SB}(n,t)$ for each block (defined below) of our codewords, and then to protect this information by enforcing two parity constraints (one on the even blocks, the other on the odd blocks). More precisely, we split the sequence $\bx$ into two sets $\bx_{e}=\left\{ \bx_{e,1},\bx_{e,2},\dotsc,\bx_{e,s}\right\}$ and $\bx_o=\left\{ \bx_{o,1},\bx_{o,2},\dotsc,\bx_{o,s+1}\right\}$, where $s=n/2P$ and $P=t2^{2t+1}\lceil \log n \rceil$:
\begin{itemize}
    \item {\bf{Even Blocks}}: $\bx_{e,i}=\bx_{[(2i-2)P+1,2iP]}, i=1,\dotsm,s$
    \item {\bf{Odd Blocks}}:
    $\bx_{o,i}=\begin{cases} \bx_{[1,P]}, &\quad i=1;\\
    \bx_{[(2i-3)P+1,(2i-1)P]}, &\quad i=2,\dotsm,s;\\
    \bx_{[n-P+1,n]}, &\quad i=s+1.
    \end{cases}$
\end{itemize}

For $i \in [s]$, let $a_{e,i}=E_{tB}(\bx_{e,i})$ and similarly let $a_{o,i}=E_{tB}(\bx_{o,i})$ for $i\in [s+1]$. Note that $\bx_e$ and $\bx_o$ each cover the sequence $\bx$ and that any interval of length $P$ is fully contained in at least one block in $\bx_e$ or in $\bx_o$. We can use the $E_{SB}$ to protect each block of length $2P$, as in the following lemma.

\begin{lemma}\label{lem:p_bounded_tburst}
There exists an integer $a$ where $a=2^{2\log P+o(\log P)}$ such that for $d_1,e_1\in[[a]]$ and $d_2,e_2\in[[a]]$, the code
\begin{equation*}
    \begin{aligned}
    \cC_{PB}(n,t,P) = \Large\{\bx\in\Sigma_2^n: &\sum_{i=1}^{s} a_{e,i}=d_1 \bmod a,
     \sum_{i=1}^{s} \left(f_t(\bx_{e,i}) \bmod a_{e,i})\right)=e_1 \bmod a,\\
    &\sum_{i=1}^{s+1} a_{o,i}=d_2 \bmod a,
    \sum_{i=1}^{s+1} \left(f_t(\bx_{o,i}) \bmod a_{o,i})\right)=e_2 \bmod a\Large\}.
    \end{aligned}
\end{equation*}
can correct a burst of at most $t$ deletions with the knowledge of the location of a substring of length $P$ from which the symbols are deleted. Furthermore, there exist choices for $d_1,d_2,e_1$ and $e_2$ such that the redundancy of the code is at most $8\log P+o(\log P)$.
\end{lemma}
\begin{proof}
The interval of length $P$ in which the edit has occurred is fully contained in a block of $\bx_e$ or in a block of $\bx_o$. Without loss of generality, let us assume the former and also assume that the index of this block is $l$. We can recover all blocks of $\bx_e$ except $\bx_{e,l}$. The value of $a_{e,l}$ and $f_{t}(\bx_{e,l})\bmod a_{e,l}$ can be determined by solving the equation $\sum_{i=1}^{s} a_{e,i}\equiv d_1 \bmod a$ and $\sum_{i=1}^{s} (f_{t}(\bx_{e,i})\bmod a_{e,i})\equiv e_1 \bmod a$, respectively. Then, by Lemma~\ref{lem:t_burst_base}, the block $\bx_{e,l}$ can be recovered.      
\end{proof}

\subsection{Overall construction of non-binary code for correcting at most $t$ deletions}\label{sec:noncort_all}

%\textcolor{red}{Not trying to create more work, but does it maybe make sense to introduce this code construction first and then prove the component codes have the desired properties. This way you can refer back to the overall code construction when you prove that the individual components $\cC_{loc}(n,c_0,c_1)$ and $cC^{b_c}(n,t,O(\log n))$ have the desired properties. We can discuss this tomorrow. }

In this subsection, we will provide the overall construction of the non-binary code for correcting at most $t$ deletions. 
\begin{theorem}
Let $q$ be an even integer. Then, there exists an integer $a=2^{2\log \log n+o(\log \log n)}$ such that for all $c_1\in[[2n]]$, $c_0\in[[4]]$, $ d_{1,i},e_{1,i}\in[[a]]$ and $d_{2,i},e_{2,i}\in[[a]]$, where $i\in[\lceil\log q\rceil]$. The code $\cC_{tB}(n)$ 
\begin{equation*}
    \cC_{tB}(n)=\{\bu\in \Sigma_q^n: A(\bu)_1\in \cC_{loc}(n,c_0,c_1),\quad
    A(\bu)_i\in \cC_{PB}(n,t,O(\log n)), \forall i\in[\lceil\log q\rceil]\}
\end{equation*}
can correct a burst of at most $t$ deletions in $q$-ary sequences. Furthermore, there exists choices for $c_0,c_1,d_{1,i},d_{2,i},e_{1,i},e_{2,i}$ such that the redundancy is at most $\log n +O(\log q\log \log n)$.
\end{theorem}

\begin{proof}
If a burst of at most $t$ deletions occurs in $\bu\in\cC_{tB}(n)$, at most $t$ adjacent columns are deleted in the binary matrix $A(\bu)$. 
Since the LSB of the binary representation of $u\in\Sigma_q$ can take any value regardless of the other bits, the pattern dense string can be applied to $A(\bu)_1$. The decoder for $\cC_{loc}(n,c_0,c_1)$ can insert the deleted bits into $A(\bu)_1$ and we can find an interval of length at most $O(\log n)$ in which the deletion has occurred. Then, $\cC_{PB}(n,t,O(\log n))$ can correct the deletions in all $\lceil\log q\rceil$ rows of $A(\bu)$ with the positional knowledge.
\end{proof}

\textit{Remark:} Although the redundancy of both non-binary codes for correcting a burst of at most 2 and $t$ deletions can be shown as $\log n+O(\log q\log\log n)$, the redundancy $8\log P+o(\log P)$ of $P$-bounded code for correcting at most $t$ deletions is higher than that of code correcting at most 2 deletions with $\log P+O(1)$. Therefore, when $t=2$, our construction in Section~\ref{sec:bicor2} is better than the construction for arbitrary $t$.

\section{Correcting a burst of at most $t$ deletions for the permutation code}\label{sec:per}

In this section, we construct a family of permutation codes that are capable of correcting a burst of at most $t$ deletions with redundancy $\log n+O(\log\log n)$. Our approach is similar in spirit to 
%Chee et al. in \cite{chee2019burst} 
previous settings considered in this paper
whereby we first attempt to approximately locate the burst of deletions, and then we seek to correct the burst. In Section~\ref{sec:per_loc}, we give a construction for the first code, which requires roughly $ \log n$ bits of redundancy, that is based upon using a simple mapping between non-binary and binary symbols. The more difficult task is to design the second code, which corrects the burst of deletions provided we roughly know its location, given that we want the second code to have redundancy of less than $\log n$ bits.

As an illustration of this difficulty, suppose we are provided with the sequence $\bpi'$, which is the result of $t$ symbols being deleted from the permutation $\bpi=(\pi_1, \ldots, \pi_n) \in \cS_n$, and suppose that we know \textit{{roughly}} where the burst of deletions has occurred, namely in an interval of length $O(\log n)$. Recall that the constraint in ~\cite{chee2019burst} only works for \emph{exactly} $t$ deletions and a set of $t$ related constraints are needed to handle \emph{at most} $t$ deletions. One naive approach is to directly split the permutation $\bpi$ into several blocks each with length $O(\log n)$ and introduce constraints into each block. However, the alphabet size for each block is still $n$, which would result in the construction of the second code that cannot have less than $\log n$ bits of redundancy.

%One naive approach is simply to apply $t$ parity constraints on the permutation $\bpi$ whereby the first constraint would require $\sum_{j=1}^{n/4} \pi_{1 + t(j-1)} \bmod n = a_1$, the second constraint would be $\sum_{j=1}^{n/4} \pi_{2 + t(j-1)} \bmod n = a_2$, and so on for some set of constants $a_1, a_2, \ldots, a_t$. The trouble with this approach is that each parity constraint involves $\log n$ bits of redundancy, which would result in a construction that has $t \log n$ bits of redundancy.

In order to avoid this issue, we introduce what is referred to as an ``overlapping ranking sequence'' for the permutation $\bpi$ in Section~\ref{sec:per_overlap}, which allows us to avoid using codes defined over large alphabets. To make use of the connection between the overlapping ranking sequence and its associated permutation, we design codes that are capable of correcting substring edits of length at most $2t$ %using the syndrome compression technique
in Section~\ref{sec:per_coroverlap}. The resulting code, which requires $O(\log \log n)$ bits of redundancy, is then shown to be capable of correcting a burst of deletions provided that we know its approximate location. The overall construction for correcting a burst of at most $t$ deletions for the permutation code and its total redundancy are presented in Section~\ref{subsec:per_all}.

%For the remainder of the section, we adopt the following notation. Let $n$ be a positive integer and $\cS_n$ be the set of all permutations on the set $[n]$. Denote $\bpi=(\pi_1,\pi_2,\dotsc,\pi_{n})\in \cS_n$ as a permutation code with length $n$. A burst of at most $t$ deletions deletes at most $t$ consecutive symbols from the permutation code $\bpi$, leading to $\bpi'=(\pi_1,\pi_2,\dotsc,\pi_i,\pi_{i+t'+1},\dotsc,\pi_n)$, where $t'\leq t$.

\subsection{Locating the deletion}\label{sec:per_loc}

In this subsection, we want to identify the location of the burst of deletions to be within an interval of size at most $O(\log n)$. Our approach will be to first convert each of the permutations in our code to be binary sequences of length $n$ by way of a simple mapping as described in Section~\ref{subsec:permu}. Afterwards, we will introduce some additional constraints (similar to the ones from Section~\ref{sec:noncort_loc}) on the resulting binary sequences that will allow us to obtain the desired localizing code.\\

\subsubsection{Mapping from permutations to binary sequences}\label{subsec:permu}
~\\ 

Define $\bb_P$: $\cS_n \rightarrow \Sigma_2^n$ as:
\begin{equation}
    b_P(\bpi)_i =\left\{\begin{array}{ll}
    1, & \text { if } \pi_i>n/2 \\
    0, & \text { if } \pi_i \leq n/2
\end{array}\right.
\end{equation}
\begin{example}
Suppose $\bpi=(5,3,4,1,2,6)$. Then, the corresponding binary sequence is $\bb_P(\bpi)=(1,0,1,0,0,1)$.
\end{example}

%\textcolor{red}{Note that in the text below, I changed the notation (similar to the previous section) to include the parameter $n$. Please make sure this is consistent with the rest of the text.}
The binary sequence $\bb_P(\bpi)$ after mapping will have an equal number of 0s and 1s when $n$ is even, and the number of 1s is one more than 0s when $n$ is odd. For even $n$, let $\cD_e(n)$ be the set of binary sequences of length $n$ with an equal number of 0s and 1s, and when $n$ is odd, let $\cD_e(n)$ be the set of binary sequences that have one more 1s than 0s. We call sequences in $\cD_e(n)$ \textit{balanced sequences}. The size of the set $|\cD_e(n)|$ is 
\begin{equation}
    |\cD_e(n)|=\left\{\begin{array}{ll}
    {n \choose n/2}, & \text{when}\; n\; \text{is even,} \\
    {n \choose (n+1)/2}, & \text{when}\; n\; \text{is odd}.
    \end{array}\right.
\end{equation}
For simplicity, we will focus on the case where $n$ is even, but similar results also hold for the case where $n$ is odd.\\

\iffalse
\begin{lemma}\label{lem:2nchoosen}
\begin{equation}
\frac{4^{n}}{\sqrt{\pi\left(n+\frac{1}{3}\right)}} \leq {2n \choose n}\leq \frac{4^{n}}{\sqrt{\pi\left(n+\frac{1}{4}\right)}}
\end{equation}
\end{lemma}
\begin{proof}
See Appendix A for details.
\end{proof}

From Lemma \ref{lem:2nchoosen}, we can get the bound of the $|\cD_e(n)|$ as
\begin{equation*}
    \frac{2^{n} \sqrt{6}}{\sqrt{\pi(3 n+2)}} \leq |\cD_e(n)| \leq \frac{2^{n+1}}{\sqrt{\pi(2 n+1)}}
\end{equation*}
\fi 

\subsubsection{Densifying binary sequences by a fixing pattern $\bw$}
~\\

Next, we make use of $(\bw, \delta)$-dense strings from the set $\cD_{\bw, \delta}(n)$ \cite{lenz2020optimal}, which was introduced in Subsection~\ref{sec:noncort_loc}. The only difference in this section is that we increase the value of $\delta$ to $\delta=t2^{2t+2}\lceil\log n\rceil$. %\textcolor{red}{I would say here that with a slight abuse of notation we assume $\delta=t2^{2t+2}\lceil\log n\rceil$ in this section. This will avoid having to sort out the $\delta'$ and $\delta$s floating around.}

\begin{lemma}\label{cor:stirling_appro}
From Stirling approximation, for all even $n\ge 2$, we have
\begin{equation*}
    \frac{2^{n} \sqrt{6}}{\sqrt{\pi(3 n+2)}} \leq\left|\mathcal{D}_{e}(n)\right|={n \choose n/2} \leq \frac{2^{n+1}}{\sqrt{\pi(2 n+1)}}.
\end{equation*}
\end{lemma}

\begin{lemma}\label{lem:size_perdense}
For all even $n\ge 2$, the number of $(\bw, \delta)$-dense strings of length $n$ among balanced sequences is
\begin{equation*}
    |\cD_e(n) \cap \cD_{\bw, \delta}(n)|\ge 
    {n \choose n/2}-\frac{2^n}{n^{2\log e-1}}\ge  \frac{2^{n} \sqrt{6}}{\sqrt{\pi(3 n+2)}}-\frac{2^n}{n^{2\log e-1}}.
\end{equation*}
\end{lemma}
\begin{proof}
Similar to the Proof of Lemma 1 in \cite{lenz2020optimal}, let $\bz \in \Sigma_2^n$ and $E_i$ be the event that $\bz_{[i+1,i+\delta]}$ does not contain the pattern $\bw$. The probability of $E_i$ is
\begin{equation}
    \pr(E_i)\leq (1-\frac{1}{2^{2t}})^{\frac{\delta}{2t}}\stackrel{(a)}\leq\frac{1}{n^{2\log e}}
\end{equation}
where (a) follows from the fact that the function $(1-1/x)^x$ is increasing in $x$ for $x>1$ and $\lim_{x\rightarrow\infty}(1-1/x)^x=1/e$. To bound the probability of the event that $\bz\in \Sigma_2^n$ is not in $\cD_{\bw,\delta}$, the union bound yields
\begin{equation}
\pr(\bz\not\in \cD_{\bw,\delta})\leq(n-\delta+1)\pr(E_i)\leq\frac{1}{n^{2\log e-1}}.
\end{equation}

Thus, %the lower bound of the number of $(\bp,\delta)$-dense strings of length $n$ in the binary sequence with equal number of 0 and 1 can be expressed as
\begin{equation}
    |\cD_e(n) \cap \cD_{\bw, \delta}(n)|\ge |\cD_e(n)|-2^n\cdot\pr(\bz\not\in \cD_{\bw,\delta})=\frac{2^{n} \sqrt{6}}{\sqrt{\pi(3 n+2)}}-\frac{2^n}{n^{2\log e-1}}.
\end{equation}
\end{proof}
Since $\frac{1}{n^{2\log e-1}}=o(\frac{\sqrt{6}}{\sqrt{\pi(3 n+2)}})$, the value of $|\cD_e(n) \cap \cD_{\bw, \delta}(n)|$ is dominated by the first term.\\

\subsubsection{Approximately locating the deletions}
%\begin{construction}
%For any integers $c_0\in[[4]]$ and %$c_1\in[[2n]]$, let
%\begin{equation*}
%    \cC_{loc}^P(n,c_0,c_1) \buildrel \Delta \over = \{ \bb_P(\bpi)\in \Sigma_2^n: \bb_P(\bpi) \in\{\cD_e(n) \cap \cD_{\bw, \delta}(n) \}, \;
    %n_{\bw}(\bb_P(\bpi)) \equiv c_0 \bmod{4}, \;
%    \VT(\balpha_{\bw}(\bb_P(\bpi)))  \equiv c_1 \bmod {2n}\}
%\end{equation*}
%\end{construction}
\begin{lemma}\label{lem:loc_permutation} For integers $c_0\in[[4]]$ and $c_1\in[[2n]]$, the code 
\begin{equation*}
    \cC_{loc}^P(n,c_0,c_1) = \{ \bx\in \Sigma_2^n: \bx \in\{\cD_e(n) \cap \cD_{\bw, \delta}(n) \}, \;
    n_{\bw}(\bx) \equiv c_0 \bmod{4}, \;
    \VT(\balpha_{\bw}(\bx)  \equiv c_1 \bmod {2n}\}
\end{equation*}
is capable of locating the burst of deletions to an interval of length at most $\delta=t2^{2t+2}\lceil\log n\rceil$.
\end{lemma}
\begin{proof}
It can be derived from Lemma~\ref{lem:loc_del}. The only difference is that the binary sequence $\bx$ should be in $\{\cD_e(n) \cap \cD_{\bw, \delta}(n)\}$.
\end{proof}

Note that if a burst of at most $t$ deletions occurs in the permutation $\bpi$, then the corresponding binary sequence $\bb_{P}(\bpi)$ suffers a burst of at most $t$ deletions at the same location. Thus, the localizing code $\cC_{loc}^P(n,c_0,c_1)$ can be used to help us determine the location of the deletions in $\bpi$.

\begin{lemma}\label{lem:size_perloc}
There exist integers $c_0$ and $c_1$ such that the size of the permutation code whose codewords $\bpi$ satisfy $\bb_P(\bpi) \in \cC_{loc}(n,c_0,c_1)$ is at least $n!/(16n)$.
\end{lemma}
\begin{proof}
By Lemma~\ref{cor:stirling_appro}, Lemma~\ref{lem:size_perdense} and the pigeonhole principle, there exist values of $c_0$ and $c_1$ such that the size of the resulting permutation code with its corresponding binary mapping sequence $\bb_P(\bpi) \in \cC_{loc}^P(n,c_0,c_1)$ is at least:
\begin{equation*}
    \frac{\left| \cD_e(n) \cap \cD_{\bw, \delta}(n) \right | \cdot \left( (n/2)! \right)^2}{4\cdot 2n} \ge\frac{\left({n\choose n/2}-\frac{2^n}{n^{2\log e-1}}\right)\cdot\frac{n!}{{n\choose n/2}}}{8n}\ge\frac{n!(1-\frac{\sqrt{6\pi}}{3}\cdot n^{1.5-2\log e})}{8n} \ge \frac{n!}{16n},
\end{equation*}
where the first inequality follows from Lemma~\ref{lem:size_perdense} and the second inequality follows from the lower bound in Lemma~\ref{cor:stirling_appro} for $n\ge 2$. The final inequality can also be shown to hold for $n\ge 3$.
\end{proof}

\iffalse
\textcolor{red}{Please put a proof of the redundancy. I agree with the derivation of the bound that holds for binary codes but the step in applying this result to the case where our codewords are permutations should also be included. In particular, notice that the size of the resulting permutation code should be the following: 
\begin{align*}
\left| D_e(n) \right | \cdot \left( (n/2)! \right)^2,    
\end{align*}
and then to compute the redundancy we need take the log of that and subtract it from $\log n!$. The result I think is actually better than advertised, since we shouldn't ``pay'' anything in terms of redundancy for having the binary image  of a permutation be balanced (it's always balanced for the way we've defined it).
}
\fi

\subsection{Mapping permutation code to the overlapping ranking sequence}\label{sec:per_overlap}

%\textcolor{red}{There's an issue here. The input to the ranking function is NOT a permutation of length $n$ when we use it to define $p_i=r(\pi_{i},\pi_{i+1},\dotsc,\pi_{i+t})$. What you basically want to do here is to project the elements $(\pi_{i},\pi_{i+1},\dotsc,\pi_{i+t})$ onto a permutation of length $t$ based upon their relative values. They should have encountered a similar problem in the Chee paper so I'd borrow their explanation/expository from there. }

In the following, we define a mapping that bears a resemblance to one originally introduced in \cite{chee2019burst} for the purpose of correcting a burst of deletions when the length of the burst was known. The key difference between their mapping and the one introduced here is that the ranking sequence in \cite{chee2019burst} was constructed using disjoint sets of symbols from the underlying permutation whereas our ranking sequence will be generated using every consecutive set of symbols from the underlying permutation. We now describe these ideas in more detail.  

For a sequence of positive integers $\bu$, the permutation projection of $\bu$, denoted by
$\mathbf{Prj}(\bu)$, is a permutation in $\cS_n$ where:
\begin{equation*}
    \prj(\bu)_{i}=\left|\left\{j: u_{j}<u_{i}, 1 \leq j \leq n\right\}\right|
    +\left|\left\{j: u_{j}=u_{i}, 1 \leq j \leq i\right\}\right|.
\end{equation*}
\begin{example} Let $\bu=\left(3,8,4,1,5\right)$. Then $\mathbf{Prj}(\bu)=\left(2,5,3,1,4\right)$.\end{example}

Let the function $\mu: \cS_{n}\rightarrow [1,n!]$ be a bijection such that $\mu(\bpi)$ is the lexicographic rank of $\bpi$ in $\cS_{n}$. For a sequence $\bu$ of length $n$, we define the permutation rank of $\bu$ as $r(\bu)=\mu(\mathbf{Prj}(\bu))\in\left[1,n!\right]$. Also, we define the corresponding overlapping ranking sequence for a permutation $\bpi$ of length $n$ as $\bp_{t+1}(\bpi)=\left(p_1,p_2,\dotsc, p_{i},\dotsc,p_{n-t}\right)\in \Sigma_{(t+1)!}^{n-t}$, where $p_i=r(\pi_{i},\pi_{i+1},\dotsc,\pi_{i+t})$. %\textcolor{red}{What you defined here with the notation `$\{ \}$' is a set and not a sequence.}
\begin{example}
For permutations of length $3$, we have $\mu(\bpi_1=(1,2,3))=1$, $\mu(\bpi_2=(1,3,2))=2$, $\mu(\bpi_3=(2,1,3))=3$, $\mu(\bpi_4=(2,3,1))=4$, $\mu(\bpi_5=(3,1,2))=5$ and $\mu(\bpi_6=(3,2,1))=6$.
\end{example}
\begin{example}
Suppose $\bpi=(6,4,2,1,5,3)$ and $t=2$, then the corresponding overlapping ranking sequence $\bp_{3}(\bpi)=(6,6,3,2)$.
\end{example}

Note that if $\bpi'$ is obtained from $\bpi$ through deletions, the identities of the deleted symbols can be determined by noting which symbols are missing. %From the previous subsection, we can successfully get the recovered $\bp_{t+1}$ at the receiver \textcolor{red}{I don't understand. The only thing the previous subsection should allow us to do is to recover the binary sequence and thus the approximate location of where the deletion occurs. I think the fact that we can get the value of the overlapping rankings sequence actually follows from the next subsection}. The only remaining problem is that whether we can recover a unique solution of $\bpi$ from $\bpi'$ and $\bp_{t+1}(\bpi)$.

%For the permutation code, since we have known the value of deleted \textcolor{red}{symbols} at the receiver, this problem can be converted to the following equivalent lemma.
\begin{lemma}\label{lem:overlappingunique}
Let $\bpi'=\left(\pi'_1,\pi'_2,\dotsc,\pi'_{n-t'}\right)$ be obtained from $\bpi \in \cS_{n}$ by deleting $t'\leq t$ consecutive symbols. Further, let $\bpi'' \in \cS_n$ be the result of inserting the deleted symbols consecutively into $\bpi'$. For any $\bpi''\neq \bpi$, the overlapping ranking sequence $\bp_{t+1}(\bpi'')$ and $\bp_{t+1}(\bpi)$ are not identical.
\end{lemma}

\begin{proof}
The lemma is proved by showing that a contradiction arises if we assume that there exist $\bpi''$ and $\bpi$ such that their corresponding overlapping ranking sequence $\bp_{t+1}(\bpi'')$ and $\bp_{t+1}(\bpi)$ are the same.

Suppose the deleted symbols from $\bpi$ are $\bpi_{[i,i+t'-1]}$ and $\bpi''$ is the result of inserting these symbols (consecutively) beginning at position $j$. Without loss generality, we assume that $j<i$. 

Thus, $\bpi$ and $\bpi''$ can be shown as:
\begin{align*}
        \bpi &=\left(\cdots, \pi_{j}, \pi_{j+1}, \cdots, \pi_{i-1}, \pi_{i}, \cdots, \pi_{i+t'-1}, \cdots  \right)\\
        \bpi'' &=\left(\cdots, \pi''_{j}, \pi''_{j+1}, \cdots, \pi''_{i-1}, \pi''_{i}, \cdots, \pi''_{i+t'-1}, \cdots  \right)
    \end{align*}

From the definition of $\bpi''$, we can have $\pi''_{k}=\pi_{k-t'}$ when $k\ge j+t'$.
To illustrate the changed and unchanged part in $\bpi$ and $\bpi''$, we denote the unchanged part in both $\bpi$ and $\bpi''$ as $(\pi_j,\pi_{j+1},\dotsc,\pi_{i-1})=(\pi''_{j+t'},\pi''_{j+t'+1},\dotsc,\pi''_{i+t'-1})=(a_1,a_2,\dotsc,a_m)$, where $m=i-j$. Further, we use $(x_1,x_2,\dotsc,x_{t'})=(\pi_i,\pi_{i+1},\dotsc,\pi_{i+t'-1})$ to denote the deleted symbols from $\bpi$ and $(x''_1,x''_2,\dotsc,x''_{t'})$ to denote the inserted symbols in $\bpi''$, where $(x''_1,x''_2,\dotsc,x''_{t'})$ is a permutation of the deleted symbols $(x_1,x_2,\dotsc,x_{t'})$.
Then, $\bpi_{[j,i+t'-1]}$ and $\bpi''_{[j,i+t'-1]}$ can be rewritten as the following, taking $m>t'$ case as example:
\begin{align*}
    \bpi &=\left(a_1, a_2, \cdots, a_{t'},a_{t'+1},\cdots, a_m, x_{1}, \cdots, x_{t'}  \right)\\
    \bpi'' &=\left(x''_{1}, x''_{2}, \cdots, x''_{t'}, a_{1},\cdots,a_{m-t'},a_{m-t'+1},\cdots, a_m \right)
\end{align*}
    
For uniformity of notation, we sometimes denote $x_j$ by $a_{m+j}$ and $x''_j$ by $a_{-t'+j}$. Let $[y_a,y_b]=\{y_a,y_{a+1},\dotsc,y_{b}\}$. For a set $\{y_1,y_2,\dotsc,y_k\}$ with distinct elements, we say $a\prec \{y_1,y_2,\dotsc,y_k\}$ if $a\le y_i$ holds for all $i\in\{1,\dotsc,k\}$, with equality holding for at most one value of $i$.

\begin{itemize}

    \item Consider the case where $i-j>t$. To guarantee each element in $\bp_{t+1}(\bpi)$ and $\bp_{t+1}(\bpi'')$ are the same, we can have the following two relationships:
    
    If $a_i\prec [a_i,a_{i+t'}]$, for $1\le i\le m-t'$, then
    \begin{equation}\label{eq:chainrelation_1}
        a_{i+t'}\prec [a_{i+t'},a_{i+2t'}]
    \end{equation}
    and if $a_i\prec [a_i,a_{m}]$, for $m-t'<i<m$, then
    \begin{equation}\label{eq:chainrelation_2}
        a_{i+t'}\prec[a_{i+t'},a_{m+t'}]
    \end{equation}
    
    Let $a^{*}=\min\{a_1,\dotsc,a_m,x_1,\dotsc,x_{t'}\}=\min\{x''_1,\dotsc,x''_{t'},a_1,\dotsc,a_m\}$. Recall that all elements are distinct.
    
    Suppose there exists some $1\le i\le m$ such that $a^{*}=a_i$. Note that elements $[a_i,a_{i+t'}]$ in $\bpi$ and $[a_{i-t'},a_i]$ in $\bpi''$ have the same value in the corresponding overlapping ranking sequence as the number of elements in each of these segments is $t'+1\le t+1$. On the other hand, if the minimum $a_i$ appears in two different places in a segment with the same starting and end location in $\bpi$ and $\bpi''$, then the overlapping ranking sequence cannot be the same. Thus, it implies the contradiction arises when considering the minimum element $a^{*}$ in $\{a_1,\dotsc,a_m\}$.
    
    Hence, there must be some $1\le i\le t'$ such that $a^{*}=x''_{i}$. Noting that $r(a_i,\dotsc,a_{i+t'})=r(x''_i,\dotsc,x''_{t'},a_1,\dotsc,a_i)$ and we have $a_{i}\prec [a_{i},a_{i+t'}]$. We now show that
    \begin{equation}
    \begin{aligned}
        a_{i}&\prec [a_{i},a_{i+t'}]\\
        a_{i+t'}&\prec [a_{i+t'},a_{i+2t'}]\\
        a_{i+2t'}&\prec [a_{i+2t'},a_{i+3t'}]\\
        &\vdots\\
        a_{i+kt'}&\prec [a_{i+kt'},a_{i+kt'+t'}]\\
        a_{i+kt'+t'}&\prec [a_{i+kt'+t'},a_{m+t'}]\\
    \end{aligned}
    \end{equation}
    where $k$ is the largest integer such that $i+kt'\le m$. All relations except the last one follow from \eqref{eq:chainrelation_1} and the last one follows from \eqref{eq:chainrelation_2}. The last two relations imply that 
    \begin{equation*}
        a_{i+kt'}\prec [a_{i+kt'},a_{m+t'}]=(a_{i+kt'},\dotsc,a_m,x_1,\dotsc,x_{t'}),
    \end{equation*}
    which is a contradiction since the minimum among all elements is among the elements $\{x''_1,\dotsc,x''_{t'}\}$, where $\{x''_1,\dotsc,x''_{t'}\}$ and $\{x_1,\dotsc,x_{t'}\}$ contain the same elements.

    \item Consider the case where $i-j\leq t$. When $j\neq i$, note that the elements $[a_{m},a_{m+t'}]$ in $\bpi$ and $[a_{m-t'},a_{m}]$ in $\bpi''$ have the same value in the corresponding overlapping ranking sequence as the number of elements in each of these segments is $t'+1\le t+1$. Also, $[a_{m},a_{m+t'}]$ in $\bpi$ and $[a_{m-t'},a_{m}]$ in $\bpi''$ have common elements $a_m$ and $x_i, \exists i\in\{1,\dotsc,t'\}$. However, this is impossible that $a_m$ and $x_i$ are in reversed order in both.
    
    When $i=j$, the elements in $\bpi_{[i,i+t]}$ and $\bpi''_{[i,i+t]}$ cannot be in the same order due to $\bpi''\neq\bpi$. 
    Thus, the overlapping ranking sequence $\bp_{t+1}$ of $\bpi$ and $\bpi''$ are not the same in this case.\qedhere
\end{itemize}
\end{proof}

After mapping the permutation code to the overlapping ranking sequence, the alphabet size of the sequence can be reduced from $n$ to $(t+1)!$. As a result, we want to correct bursts of deletions in $\bpi$ by first recovering the corresponding overlapping ranking sequence $\bp_{t+1}(\bpi)$, and then we will use this information to uniquely determine $\bpi$ according to Lemma~\ref{lem:overlappingunique}. Recall that for a string $(v_1, v_2, \ldots, v_{n})$, we say that $(v_{i}, v_{i+1}, \ldots, v_{i+\ell-1})$ is a substring of length $\ell$ that appears in $(v_1, v_2, \ldots, v_{n})$ at position $i$.

%\textcolor{red}{Shuche, let's go with this characterization of the errors in the overlapping rankings sequence. I think it's much cleaner. Please clean up the next subsection using this characterization. Should be easier to characterize the error ball as well. I'd stick with the notation $\cB_{2t}$ as well. You can refer to this as a substring error-correcting code.  }

\begin{claim}\label{cl:rankconf}
After deleting a burst of at most $t$ symbols in a permutation $\bpi$ resulting in $\bpi'$, the corresponding overlapping ranking sequence $\bp_{t+1}(\bpi')$ can be obtained from $\bp_{t+1}(\bpi)$ by at most $t$ consecutive substitutions followed by a burst of at most $t$ deletions.
\end{claim}

\begin{proof}
We can write $\bpi$ and $\bp_{t+1}(\bpi)$ as:
\begin{align}\label{eq:overlap_ori}
        \bpi &=\left(\pi_{1}, \pi_{2}, \cdots, \pi_{i-t}, \cdots ,\pi_{i-1},\pi_{i}, \pi_{i+1}, \cdots, \pi_{i+t'-1}, \pi_{i+t'}, \cdots  \right)\notag\\
        \bp_{t+1}(\bpi) &=\left(p_{1}, p_{2}, \cdots, p_{i-t}, \cdots ,p_{i-1},p_{i}, p_{i+1}, \cdots, p_{i+t'-1}, p_{i+t'}, \cdots   \right)
    \end{align}
Let $\bpi'=\left(\pi'_1,\pi'_2,\dotsc,\pi'_{n-t'}\right)$ be the result of deleting $t'\leq t$ consecutive symbols from the permutation $\bpi$ and suppose the deleted symbols from $\bpi$ are $\bpi_{[i,i+t'-1]}$. 

Thus, $\bpi'$ and $\bp_{t+1}(\bpi')$ can be written as:
\begin{align}\label{eq:overlap_del}
        \bpi' &=\left(\pi_{1}, \pi_{2}, \cdots, \pi_{i-t}, \cdots, \pi_{i-1}, \pi_{i+t'}, \cdots  \right)\notag\\
        \bp_{t+1}(\bpi') &=\left(p_{1}, p_{2}, \cdots, p'_{i-t}, \cdots, p'_{i-1}, p_{i+t'}, \cdots   \right)
    \end{align}
where we use $p_i$ to denote an unchanged value and $p'_j$ to denote a possibly changed value in $\bp_{t+1}(\bpi')$ compared with $\bp_{t+1}(\bpi)$.

By comparing \eqref{eq:overlap_del} with \eqref{eq:overlap_ori}, we see that there are at most $t$ consecutive substitutions (substituting $(p_{i-t},\dotsc,p_{i-1})$ by $(p'_{i-t},\dotsc,p'_{i-1})$) followed by at most $t$ consecutive deletions (deleting $(p_{i},\dotsc,p_{\textcolor{red}{i-t'-1}})$).
\end{proof}

Based on this observation, we characterize this error pattern as substring edits that replace a substring of length at most $2t$ with another substring of length at most $2t$, which is a more general type of error. Thus, in the next subsection, we will discuss how to construct codes capable of correcting substring edits of length at most $2t$ for recovering the overlapping ranking sequence $\bp_{t+1}(\bpi)$.

%Then, we generalize this claim as after deleting a burst of at most $t$ symbols in a permutation $\bpi$, there will be at most $2t$ consecutive edits in $\bp_{t+1}(\bpi)$. Motivated by this observation,  we will discuss how to correct at most $2t$ consecutive edits in $\bp_{t+1}(\bpi)$ in the next subsection.

\subsection{Correcting substring edits of length at most $2t$ in the overlapping ranking sequence}\label{sec:per_coroverlap}

%\textcolor{red}{We need to link/relate this approach to the one used in the previous section, since it's more or less the same. I don't know that we need to re-introduce everything in this way again. For lemma 13 we should say why we're introducing it and why we can' use Lemma 6 from before. At the very least I don't think that we need to have proofs. I would just say we're going to repeat the same idea as before (if it's different in ANY way say how and remove proof of Lemma 15.}

In this section, 
%we denote the error pattern of replacing a substring of length at most $2t$ with another substring of length at most $2t$ as substring edits of length at most $2t$. Then, 
our goal is to construct a code for correcting substring edits of length at most $2t$ in the overlapping ranking sequence $\bp_{t+1}(\bpi)$ based on the systematic binary code capable of correcting up to $t$ edits \cite{sima2020optimalbinary}, where each edit is a deletion, insertion or substitution error.

For $q<\log n$, the basic idea is to consider $q$-ary sequences as a set of $\lceil \log q\rceil$ binary sequences. Unlike the setup in Section~\ref{sec:noncort} where we only had to correct deletions, for our current setup we want to correct deletions and substitutions. It is straightforward to see that the number of edits for substring edits of length at most $2t$ is also at most $2t$. Thus, we should set the number of edits to $2t$ in our problem.

\begin{lemma}\label{lem:binary_2t}
(c.f., \cite{sima2020optimalbinary}) Let $t$ be a constant with respect to $k$. There exist an integer $a\leq 2^{4t\log k+o(\log k)}$ and a labeling function $f_{2t}:\Sigma_2^k\rightarrow\Sigma_{2^{\cR_{2t}(k)}}$, where $\cR_{2t}(k)=O(t^4\log k)$ such that $\{(\bx,a,f_{2t}(\bx)\bmod a):\bx\in\Sigma_2^k\}$ can correct substring edits of length at most $2t$. 
\end{lemma}

To extend this base code to nonbinary, we can apply the code in Lemma~\ref{lem:binary_2t} to each row of $A(\bu)$ for $\bu\in\Sigma_q^k$. Therefore, we can get the following lemma for $q$-ary sequences.
\begin{lemma}\label{lem:qaryhash}
Let $t$ be a constant with respect to $k$. There exist an integer $a_q\leq 2^{\lceil\log q \rceil(4t\log k+o(\log k))}$ and a labeling function $f_{2t}^q:\Sigma_q^k\rightarrow\Sigma_{2^{\lceil\log q\rceil\cR_{2t}(k)}}$, where $f_{2t}^q(\bu)=\sum_{i=1}^{\lceil\log q\rceil} 2^{\cR_{2t}(k)(i-1)}f_{2t}(A(\bu)_i)$ such that $\{(\bu,a_q,f_{2t}^q(\bu)\bmod a_q):\bu\in\Sigma_q^k\}$ can correct substring edits of length at most $2t$ in $q$-ary sequences.
\end{lemma}

From Lemma~\ref{lem:loc_permutation}, we can narrow the deletion to an $O(\log n)$ interval in the permutation $\bpi$. Then, we will make use of Lemma~\ref{lem:qaryhash} to construct a code for correcting substring edits of length at most $2t$ in the corresponding overlapping ranking sequence $\bp_{t+1}(\bpi)$ with this positional knowledge (We omit the argument $t+1$ and $\bpi$ from $\bp_{t+1}(\bpi)$ and simply write $\bp$ in the rest of this subsection). 

We split the sequence $\bp$ into two sets $\bp_{e}=\left\{ \bp_{e,1},\bp_{e,2},\dotsc,\bp_{e,s}\right\}$ and $\bp_o=\left\{ \bp_{o,1},\bp_{o,2},\dotsc,\bp_{o,s+1}\right\}$, where $s=n/2P$ and $P=t2^{2t+2}\lceil\log n\rceil$ for even and odd blocks, respectively, which is the same as the manner in Section~\ref{sec:noncort_p}. Similarly, we can use the syndrome $(a_q,f_{2t}^q(\bu)\bmod a_q)$ to protect each block of length $2P$, as in the following lemma. All of notations are analogous to those from Lemma~\ref{lem:p_bounded_tburst} except that $\bx$ is replaced with $\bp$ and $a_{e,i}/a_{o,i}$ are replaced with $a^q_{e,i}/a^q_{o,i}$, where $a^q_{e,i}=E_{tB}(\bp_{e,i})$ for $i\in [s]$ and similarly $a^q_{o,i}=E_{tB}(\bp_{o,i})$ for $i \in [s+1]$.

\iffalse
First, we spilt the sequence $\bp$ into two subsequences, $\bp_{e}=\left\{ \bp_{e,1},\bp_{e,2},\dotsc,\bp_{e,s}\right\}$ and $\bp_o=\left\{ \bp_{o,1},\bp_{o,2},\dotsc,\bp_{o,s}\right\}$, where $s=n/2P$ and $P=t2^{2t+2}\lceil\log n\rceil$:
\begin{itemize}
    \item {\bf{Even Blocks}}: $\bp_{e,i}=\bp_{[(2i-2)P+1,2iP]},\   i=1,\dotsm,s$
    \item {\bf{Odd Blocks}}:\\
    $\bp_{o,i}=\begin{cases} \bp_{[1,P]}, &\ i=1;\\
    \bp_{[(2i-1)P+1,(2i+1)P]}, &\ i=2,\dotsm,s-1;\\
    \bp_{[n-P+1,n]}, &\ i=s.
    \end{cases}$
\end{itemize}
\fi

%Note that $\bp_e$ and $\bp_o$ each cover the sequence $\bp$ and that any interval of length $P$ is fully contained in at least one block in $\bp_e$ or in $\bp_o$. %After splitting the overlapping ranking sequence into even and odd blocks, 

\begin{lemma}\label{lem:correctoverlap}
There exists an integer $a= 2^{\lceil\log q\rceil(4t\log (2P)+o(\log P))}$ such that for any $d_1,e_1\in[[a]]$, $d_2,e_2\in[[a]]$, the code
\begin{equation*}
    \begin{aligned}
    \cC_{2t}(n,t,P) = \Large\{\bp\in\Sigma_{(t+1)!}^n: &\sum_{i=1}^{s} a^q_{e,i}=d_1 \bmod a,
     \sum_{i=1}^{s} \left(f_{2t}^q(\bp_{e,i}) \bmod a^q_{e,i})\right)=e_1 \bmod a,\\
    &\sum_{i=1}^{s+1} a^q_{o,i}=d_2 \bmod a,
    \sum_{i=1}^{s+1} \left(f_{2t}^q(\bp_{o,i}) \bmod a^q_{o,i})\right)=e_2 \bmod a\Large\}.
    \end{aligned}
\end{equation*}
can correct one substring edit of length at most $2t$ with the knowledge that the location of the edited symbols is within $P$ consecutive positions. Furthermore, there exist choices for $d_1,d_2$ and $e_1,e_2$ such that the redundancy of the code is at most $4\log a$.
\end{lemma}

\iffalse
\begin{proof}
The interval of length $P$ in which the edit has occurred is fully contained in a block of $\bp_e$ or in a block of $\bp_o$. Without loss of generality, let us assume the former and also assume that the index of this block is $l$. We can recover all blocks of $\bp_e$ except $\bp_{e,l}$. The value of $f_{2t}^q(\bp_{e,l})$ can be determined by solving the equation $\sum_{i=1}^{s} f_{2t}^q(\bp_{e,i})\equiv d_1 \bmod a$. Then, by Lemma~\ref{lem:qaryhash}, the block $\bp_{e,l}$ can be recovered.  
\end{proof}
\fi
%Since the deletion position has been located to be in a an interval of length $P=\cO(\log n)$, blocks $\bp_{e,j}$ and $\bp_{o,j}$ for all $j=1,\dotsm,s$ can be recovered except two possible blocks indexed $l$ in even or odd blocks, and then we can get the corresponding syndrome of $\bp_{e,l}$ or $\bp_{o,l}$. Therefore, suppose the deletion occurred in the block $\bp_{e,l}$, which can be recovered through the syndrome of $\left(f_{2t}^q(\bp_{e,l})\bmod a\right)$. 

\subsection{Overall construction}\label{subsec:per_all}

Building on the previous sections, we can present the overall construction of the permutation code for correcting a burst of at most $t$ deletions. First, we apply the code $\cC_{loc}^P(n,c_0,c_1)$ to narrow the deletion into an interval of length $t2^{2t+2}\lceil\log n\rceil$ with redundancy $\log n+O(1)$. Then, we recover the permutation via $\cC_{2t}(n,(t+1)!, t2^{2t+2}\lceil\log n\rceil)$ for correcting the corresponding overlapping ranking sequence with the positional knowledge of the deletion.

\iffalse
\begin{construction}
For arbitrary integers $n$, $a$, $c_0$, $c_1$, $d_1$ and $d_2$, we define a  as
\begin{equation*}
    \cP_{\leq t}(n) = \left \{ \bpi\in \cS_n: \bb_P(\bpi) \in \cC_{loc}^P(n,c_0,c_1),\quad
    \bp_{t+1}(\bpi)\in \cC_{2t}(n,(t+1)!,  t2^{2t+2}\lceil\log n\rceil) \right\}
\end{equation*}
\end{construction}
\fi

\begin{theorem}\label{eq:permumain}
There exists an integer $a=2^{\lceil\log (t+1)!\rceil(4t\log \log n+o(\log \log n))}$ such that for all $c_0\in[[4]]$, $c_1\in[[2n]]$, $ d_1,d_2\in[[a]]$ and $e_1,e_2\in[[a]]$. The permutation code $\cP_{\leq t}(n)$ over $\cS_n$
\begin{equation*}
    \cP_{\leq t}(n) = \left \{ \bpi\in \cS_n: \bb_P(\bpi) \in \cC_{loc}^P(n,c_0,c_1),\quad
    \bp_{t+1}(\bpi)\in \cC_{2t}(n,(t+1)!,  t2^{2t+2}\lceil\log n\rceil) \right\}
\end{equation*}
is capable of correcting a burst of at most $t$ deletions with the redundancy at most $\log n+O(\log\log n)$ bits.
\end{theorem}
\begin{proof}
The error-correcting capability of the code has already been discussed. From Lemma~\ref{lem:correctoverlap}, the redundancy of the second part in the permutation code $\cP_{\leq t}(n)$ for correcting the overlapping ranking sequence will be $4\log a$. Since $P=t2^{2t+2}\lceil\log n\rceil$ and $t$ is a constant, we have
\begin{equation*}
    4\log a=O(\log\log n).
\end{equation*}
Combining with Lemma~\ref{lem:size_perloc}, the code size $|\cP_{\leq t}(n)|$ is at least
\begin{equation*}
    |\cP_{\leq t}(n)|=\frac{n!}{16n\cdot a^4}\geq\frac{n!}{16n\cdot 2^{O(\log\log n)}}.
\end{equation*}
Therefore, the total redundancy of the permutation code $\cP_{\leq t}(n)$ is at most $\log n+O(\log\log n)$.
\end{proof}

\section{Conclusion}\label{sec:conclusion}
Motivated by applications to DNA storage, we have constructed non-binary codes capable of correcting bursts of deletions. By considering a variation of the well-known Levenshtein code, we presented a non-binary code capable of correcting bursts of length at most $2$. We then developed codes capable of correcting bursts of length at most $t$ before turning our attention to burst-error-correcting codes for permutations. Each of the proposed codes in this paper is nearly optimal in terms of the number of redundant bits.

Although in many cases our results improve upon the prior art, there are many avenues for future research:
\begin{itemize}
    \item \textit{\textbf{Systematic $t$-burst-error-correcting codes}}: Although the non-binary codes presented in this work were nearly optimal in terms of their redundancy, the proposed codes were non-systematic. As discussed in Section~\ref{sec:noncor2}, even for the case of $t=2$, the authors are unaware of a systematic encoding that approaches our results.
    \item \textit{\textbf{Codes correcting bursts of edits}}:
    Codes that correct bursts of insertions/deletions/substitutions have applications not only in DNA storage but also in other areas such as in the document exchange problem. Currently, no optimal constructions for this setup have been reported in the open literature.
    \item  \textit{\textbf{Codes correcting multiple bursts of deletions}}: Even for the case of $2$ bursts of deletions, there are many different problems of interest. One could consider the setup where the bursts are each of the same length or possibly of different lengths. Additionally, the problem of constructing codes correcting multiple bursts of insertions/deletions/substitutions is another area of future work.
\end{itemize}

\bibliographystyle{IEEEtran}
\bibliography{references}

% Generated by IEEEtran.bst, version: 1.14 (2015/08/26)
\begin{thebibliography}{10}
\providecommand{\url}[1]{#1}
\csname url@samestyle\endcsname
\providecommand{\newblock}{\relax}
\providecommand{\bibinfo}[2]{#2}
\providecommand{\BIBentrySTDinterwordspacing}{\spaceskip=0pt\relax}
\providecommand{\BIBentryALTinterwordstretchfactor}{4}
\providecommand{\BIBentryALTinterwordspacing}{\spaceskip=\fontdimen2\font plus
\BIBentryALTinterwordstretchfactor\fontdimen3\font minus
  \fontdimen4\font\relax}
\providecommand{\BIBforeignlanguage}[2]{{%
\expandafter\ifx\csname l@#1\endcsname\relax
\typeout{** WARNING: IEEEtran.bst: No hyphenation pattern has been}%
\typeout{** loaded for the language `#1'. Using the pattern for}%
\typeout{** the default language instead.}%
\else
\language=\csname l@#1\endcsname
\fi
#2}}
\providecommand{\BIBdecl}{\relax}
\BIBdecl

\bibitem{wang2021non}
S.~Wang, J.~Sima, and F.~Farnoud, ``Non-binary codes for correcting a burst of
  at most 2 deletions,'' in \emph{2021 IEEE International Symposium on
  Information Theory (ISIT)}.\hskip 1em plus 0.5em minus 0.4em\relax IEEE,
  2021, pp. 2804--2809.

\bibitem{Wang2022permutation}
S.~Wang, Y.~Tang, R.~Gabrys, and F.~Farnoud, ``Permutation codes for correcting
  a burst of at most $t$ deletions,'' in \emph{58th Allerton Conference on
  Communication, Control, and Computing}, vol.~1, 2022, pp. 1--6.

\bibitem{chee2019burst}
Y.~M. Chee, S.~Ling, T.~T. Nguyen, H.~Wei, X.~Zhang \emph{et~al.},
  ``Burst-deletion-correcting codes for permutations and multipermutations,''
  \emph{IEEE Transactions on Information Theory}, vol.~66, no.~2, pp. 957--969,
  2019.

\bibitem{yazdi2015dna}
S.~H.~T. Yazdi, H.~M. Kiah, E.~Garcia-Ruiz, J.~Ma, H.~Zhao, and O.~Milenkovic,
  ``{{DNA}}-based storage: Trends and methods,'' \emph{IEEE Transactions on
  Molecular, Biological and Multi-Scale Communications}, vol.~1, no.~3, pp.
  230--248, 2015.

\bibitem{dolecek2007using}
L.~Dolecek and V.~Anantharam, ``Using {{Reed–Muller}} {{RM}}$(1, m)$ codes
  over channels with synchronization and substitution errors,'' \emph{IEEE
  Transactions on Information Theory}, vol.~53, no.~4, pp. 1430--1443, 2007.

\bibitem{venkataramanan2010interactive}
R.~Venkataramanan, H.~Zhang, and K.~Ramchandran, ``Interactive low-complexity
  codes for synchronization from deletions and insertions,'' in \emph{2010 48th
  Annual Allerton Conference on Communication, Control, and Computing
  (Allerton)}.\hskip 1em plus 0.5em minus 0.4em\relax IEEE, 2010, pp.
  1412--1419.

\bibitem{levenshtein1966binary}
V.~I. Levenshtein \emph{et~al.}, ``Binary codes capable of correcting
  deletions, insertions, and reversals,'' in \emph{Soviet physics doklady},
  vol.~10, no.~8.\hskip 1em plus 0.5em minus 0.4em\relax Soviet Union, 1966,
  pp. 707--710.

\bibitem{lee2019terminator}
H.~H. Lee, R.~Kalhor, N.~Goela, J.~Bolot, and G.~M. Church, ``Terminator-free
  template-independent enzymatic {{DNA}} synthesis for digital information
  storage,'' \emph{Nature communications}, vol.~10, no.~1, pp. 1--12, 2019.

\bibitem{levenshtein1967asymptotically}
V.~Levenshtein, ``Asymptotically optimum binary code with correction for losses
  of one or two adjacent bits,'' \emph{Problemy Kibernetiki}, vol.~19, pp.
  293--298, 1967.

\bibitem{schoeny2017codes}
C.~Schoeny, A.~Wachter-Zeh, R.~Gabrys, and E.~Yaakobi, ``Codes correcting a
  burst of deletions or insertions,'' \emph{IEEE Transactions on Information
  Theory}, vol.~63, no.~4, pp. 1971--1985, 2017.

\bibitem{lenz2020optimal}
A.~Lenz and N.~Polyanskii, ``Optimal codes correcting a burst of deletions of
  variable length,'' in \emph{International Symposium on Information Theory
  (ISIT)}, 2020.

\bibitem{sima2020syndrome}
J.~Sima, R.~Gabrys, and J.~Bruck, ``Syndrome compression for optimal redundancy
  codes,'' in \emph{2020 IEEE International Symposium on Information Theory
  (ISIT)}.\hskip 1em plus 0.5em minus 0.4em\relax IEEE, 2020, pp. 751--756.

\bibitem{sun2022improved}
Y.~Sun, Y.~Zhang, and G.~Ge, ``Improved constructions of permutation and
  multi-permutation codes correcting a burst of stable deletions,'' \emph{arXiv
  preprint arXiv:2208.10110}, 2022.

\bibitem{schoeny2017novel}
C.~Schoeny, F.~Sala, and L.~Dolecek, ``Novel combinatorial coding results for
  {{DNA}} sequencing and data storage,'' in \emph{2017 51st Asilomar Conference
  on Signals, Systems, and Computers}.\hskip 1em plus 0.5em minus 0.4em\relax
  IEEE, 2017, pp. 511--515.

\bibitem{sloane2000single}
N.~J. Sloane, ``On single-deletion-correcting codes,'' \emph{Codes and
  designs}, vol.~10, pp. 273--291, 2000.

\bibitem{tenengolts1984nonbinary}
G.~Tenengolts, ``Nonbinary codes, correcting single deletion or insertion
  (corresp.),'' \emph{IEEE Transactions on Information Theory}, vol.~30, no.~5,
  pp. 766--769, 1984.

\bibitem{chee2018coding}
Y.~M. Chee, H.~M. Kiah, A.~Vardy, E.~Yaakobi \emph{et~al.}, ``Coding for
  racetrack memories,'' \emph{IEEE Transactions on Information Theory},
  vol.~64, no.~11, pp. 7094--7112, 2018.

\bibitem{sima2020optimalbinary}
J.~Sima, R.~Gabrys, and J.~Bruck, ``Optimal systematic t-deletion correcting
  codes,'' in \emph{2020 IEEE International Symposium on Information Theory
  (ISIT)}.\hskip 1em plus 0.5em minus 0.4em\relax IEEE, 2020, pp. 769--774.

\bibitem{kulkarni2013nonasymptotic}
A.~A. Kulkarni and N.~Kiyavash, ``Nonasymptotic upper bounds for deletion
  correcting codes,'' \emph{IEEE Transactions on Information Theory}, vol.~59,
  no.~8, pp. 5115--5130, 2013.

\end{thebibliography}

\begin{appendices}

\iffalse
\section{Proof of Lemma \ref{lem:2nchoosen}}
\begin{proof}

(of Lemma \ref{lem:2nchoosen}) First, we have 
\begin{equation*}
\left(\frac{n+\frac{1}{2}}{n+1}\right)^{2} =\frac{n^{2}+n+\frac{1}{4}}{n^{2}+2 n+1} \leq \frac{n+\frac{1}{3}}{n+\frac{4}{3}}
\end{equation*}

Thus,
\begin{equation*}
\frac{{{2n+2}\choose{n+1}}}{{2n\choose n}} =4 \frac{n+\frac{1}{2}}{n+1} 
 \leq 4 \sqrt{\frac{n+\frac{1}{3}}{n+\frac{4}{3}}}
\end{equation*}
which implies ${2n\choose n}\frac{n+1/3}{4^n}$ is decreasing.

We also have
\begin{equation*}
   \left(\frac{n+\frac{1}{2}}{n+1}\right)^{2} =\frac{n^{2}+n+\frac{1}{4}}{n^{2}+2 n+1} \leq \frac{n+\frac{1}{4}}{n+\frac{5}{4}}
\end{equation*}
Then,
\begin{equation*}
\frac{{{2n+2}\choose{n+1}}}{{2n\choose n}} =4 \frac{n+\frac{1}{2}}{n+1} 
 \leq 4 \sqrt{\frac{n+\frac{1}{4}}{n+\frac{5}{4}}}
\end{equation*}
which implies ${2n\choose n}\frac{n+1/4}{4^n}$ is increasing.

From 
\begin{equation*}
    \lim _{n \rightarrow \infty} \frac{\sqrt{\pi n}}{4^{n}}{2n \choose n}=1
\end{equation*}

We can show that
\begin{equation*}
\frac{4^{n}}{\sqrt{\pi\left(n+\frac{1}{3}\right)}} \leq {2n \choose n}\leq \frac{4^{n}}{\sqrt{\pi\left(n+\frac{1}{4}\right)}}
\end{equation*}
\end{proof}
\fi

\section{Proof of Theorem \ref{thm:nonasymptotic}}
In order to make use of this technique, we need a few results related to the set $D_t(\bu)$, which appear as Claims~\ref{cl:dballdec} and \ref{cl:dballsize}. In the following, let $N(n,t,i) = \left| \left \{ \bu \in \Sigma_q^n : |D_t(\bu)| = i \right \} \right|$. 

\begin{claim}\label{cl:dballdec} (c.f.,\cite[Lemma 1]{schoeny2017codes}) Let $\bu \in \Sigma_q^n$ and suppose $\bu' \in D_t(\bu)$. Then $|D_t(\bu)| \geq |D_t(\bu')|$.
\end{claim}

\begin{claim}\label{cl:dballsize} For $i \in [n-t+1]$ and $t | n$, we have that
$$ N(n,t,i) = q^t (q-1)^{i-1} \binom{n-t}{i-1}.$$
\end{claim}
\begin{proof}
We can arrange sequence $\bu=(u_1,u_2,\dotsc,u_n)\in \Sigma_q^n$ into a $t\times \frac{n}{t}$ array $A_q(\bu)$ as the following
    \begin{equation*}
        A_q(\bu)=\left[ \begin{matrix}
{u_{1}}&{u_{t+1}}& \cdots &{u_{n-t+1}}\\
{u_{2}}&{u_{t+2}}& \cdots &{u_{n-t+2}}\\
\vdots&\vdots&\vdots&\vdots\\
{u_{t}}&{u_{2t}}& \cdots &{u_{n}}
\end{matrix} \right]
\end{equation*}

Let $N_r(\bu_j)$ denote the number of runs in the $j$th row of $A_q(\bu)$. The size of $t$-burst-deletion ball $|D_t(\bu)|$ was shown in \cite{levenshtein1966binary}
\begin{equation}\label{eq:DtNr}
    |D_t(\bu)|=\left(\sum_{j=1}^{t}N_r(\bu_j)\right)-t+1.
\end{equation}

Then, the problem of counting the number of $q$-ary sequences with length $n$ whose $t$-burst-deletion ball is $i$ is equivalent to counting the number of $q$-ary sequences with length $n$ for which
\begin{equation}\label{eq:Nrit}
    \left(\sum_{j=1}^{t}N_r(\bu_j)\right)=i+t-1.
\end{equation}
The number of $q$-ary sequences of length $n$ with $r$ runs is
\begin{equation}\label{eq:Mnr}
    M(n,r) := q(q-1)^{r-1}{n-1 \choose r-1}. 
\end{equation}

Combining (\ref{eq:DtNr}), (\ref{eq:Nrit}) and (\ref{eq:Mnr}), we have 
\begin{align*}
    N(n,t,i) =& \sum_{r_1 + r_2 + \cdots + r_t=i+t-1} M\left(\frac{n}{t},r_1\right) M\left(\frac{n}{t},r_2\right) \cdots M\left(\frac{n}{t},r_t\right) \\
    =& q^t (q-1)^{i-1} \sum_{r_1 + \cdots + r_t=i+t-1} \binom{\frac{n}{t}-1}{r_1-1} \binom{\frac{n}{t}-1}{r_2-1} \cdots \binom{\frac{n}{t}-1}{r_t-1} \\
    =& q^t (q-1)^{i-1} \sum_{r_1 + \cdots + r_t=i+t-1} \binom{\frac{n}{t}-1}{\frac{n}{t}-r_1} \binom{\frac{n}{t}-1}{\frac{n}{t}-r_2} \cdots \binom{\frac{n}{t}-1}{\frac{n}{t}-r_t}.
    %=& q^t (q-1)^{i-1} \sum_{r_1 + \cdots + r_t=i+t-1} \binom{\frac{n}{t}-1}{\frac{n}{t}-r_1} \binom{\frac{n}{t}-1}{\frac{n}{t}-r_2} \cdots \binom{\frac{n}{t}-1}{\frac{n}{t}-r_t}.
\end{align*}
Since $\sum_{j=1}^t \frac{n}{t} - r_j = n - \left( i + t - 1 \right)$ and $\sum_{j=1}^t \frac{n}{t} - 1 = n -t$, using a generalized Vandermonde identity we have
\begin{align*}
    N(n,t,i) =& q^t (q-1)^{i-1} \binom{n-t}{n-i-t+1} \\
    =& q^t (q-1)^{i-1} \binom{n-t}{i-1}.\qedhere
\end{align*}
\end{proof}

We now proceed to the proof of Theorem~\ref{thm:nonasymptotic}.

\noindent \textit{Proof of Theorem~\ref{thm:nonasymptotic}:}
We proceed similarly to the method presented in \cite{schoeny2017codes}. Let $\cH_{q,t,n}$ be the following hypergraph:
\begin{equation*}
    \cH_{q,t,n}=\left( \Sigma_q^{n-t}, \left \{D_t\left(\bu \right):\bu\in\Sigma_q^n \right \} \right).
\end{equation*}

It is known \cite{kulkarni2013nonasymptotic} that under this setup $|\cM_t(n)| \leq \tau^{*}(\cH_{q,t,n})$ where $\tau^{*}(\cH_{q,t,n})$ is the solution to the following linear program:
\begin{align}
\tau^{*}\left(\mathcal{H}_{q, t, n}\right) =\min &\left\{\sum_{\bu' \in \Sigma_{q}^{n-t}} w(\bu')\right\} \label{eq:lpb1} \\
\text { s.t. } &\sum_{\bu' \in D_{t}(\bu)} w(\bu') \geq 1, \forall \bu \in \Sigma_{q}^{n} \label{eq:lpc1} \\
&w(\bu') \geq 0, \forall \bu' \in \Sigma_{q}^{n-t}. \label{eq:lpc2}
\end{align}

Let $w : \Sigma_q^{n-t} \to \mathbb{R}$ be defined such that $w(\bu')=\frac{1}{|D_t(\bu')|}, \forall \bu' \in \Sigma_{q}^{n-t}$. Clearly $w(\bu') \geq 0$ for any $\bu' \in \Sigma_q^{n-t}$. As a result of Claim~\ref{cl:dballdec}, we have
\begin{align*}
    \sum_{\bu' \in D_{t}(\bu)} w(\bu') = \sum_{\bu' \in D_{t}(\bu)} \frac{1}{|D_t(\bu')|} \geq \sum_{\bu' \in D_{t}(\bu)}\frac{1}{|D_t(\bu)|} = 1,
\end{align*}
so that the function $w$ satisfies both (\ref{eq:lpc1}) and (\ref{eq:lpc2}).

Then, according to (\ref{eq:lpb1})
\begin{equation*}
|\cM_t(n)| \leq \sum_{\bu' \in D_{t}(\bu)} \frac{1}{|D_t(\bu')|}.
\end{equation*}

For $1\le i\le n-t+1$, denote $N(n,t,i)$ as the size of the set $\{\bu\in \Sigma_q^n: |D_t(\bu)|=i\}$, where $N(n,t,i)=q^t(q-1)^{i-1}\binom{n-t}{i-1}$.

\begin{equation}
\begin{aligned}
\sum_{\bu' \in D_{t}(\bu)} w(\bu')&=\sum_{\bu' \in D_{t}(\bu)} \frac{1}{|D_t(\bu')|}\\
&=\sum_{i=1}^{n-2t+1}\frac{N(n-t,t,i)}{i}\\
&=q^t\sum_{i=1}^{n-2t+1} (q-1)^{i-1}\frac{{n-2t \choose i-1}}{i} \nonumber \\
&=q^t\sum_{i=1}^{n-2t+1} (q-1)^{i-1}\frac{(n-2t)!}{i!(n-2t-i+1)!} \nonumber \\
&=\frac{q^t}{(q-1)(n-2t+1)}\sum_{i=1}^{n-2t+1} (q-1)^{i}\frac{(n-2t+1)!}{i!(n-2t-i+1)!}\\
&=\frac{q^t}{(q-1)(n-2t+1)}\sum_{i=1}^{n-2t+1} (q-1)^{i}{n-2t+1 \choose i}\\
&\stackrel{(a)}=\frac{q^t}{(q-1)(n-2t+1)}(q^{n-2t+1}-1)\\
&=\frac{q^{n-t+1}-q^t}{(q-1)(n-2t+1)}.
\end{aligned}
\end{equation}
where (a) follows from the Binomial theorem. 
\qedsymbol{}

\end{appendices}

\end{document}